\documentclass[letterpaper,12pt,twoside]{article}
\usepackage{amsfonts}
\usepackage{amsmath}
\usepackage{amsopn}
\usepackage{amsthm}
\usepackage{amssymb}
\usepackage{graphicx}
\usepackage{rotating}
\usepackage[]{natbib}
\usepackage[english]{babel}
\usepackage[utf8]{inputenc}
\usepackage[colorlinks=true,citecolor=blue,linkcolor=blue]{hyperref}
\usepackage{xspace}
\usepackage{setspace}
\usepackage{epstopdf}
\usepackage{dsfont}
\usepackage{booktabs}
\usepackage[T1]{fontenc}
\usepackage{anysize}
\usepackage{enumerate}
\usepackage[labelfont=bf, labelsep=period]{caption}
\marginsize{1.87cm}{1.87cm}{1.87cm}{1.87cm} % left, right, top , bottom
\usepackage{fancyhdr}
\usepackage{floatrow}
\newfloatcommand{capbtabbox}{table}[][\FBwidth]
\usepackage{algorithm}% http://ctan.org/pkg/algorithms
\RequirePackage{txfonts}
\usepackage{times}
\usepackage{bm}
\usepackage{soul}
\usepackage[linkcolor=blue,colorlinks=true]{hyperref}

\newtheorem{Proposition}{Proposition}

\newtheorem{Theorem}{Theorem}
\newtheorem{Lemma}{Lemma}

\newtheorem{Assumption}{Assumption}
\DeclareMathOperator{\re}{\mathbb{R}}
\DeclareMathOperator{\na}{\mathbb{N}}

\newcommand{\E}{\mathbb{E}}

\newcommand{\ind}{\mathds{1}}

\newcommand{\Prob}{\mathbb{P}}
\newcommand{\F}{\mathcal{F}}

\renewcommand{\qedsymbol}{$\blacksquare$}
\newcommand{\neigh}{\mathrm{N}}

\newcommand{\Xset}{\mathcal{X}}
\newcommand{\Xalg}{\mathsf{X}}

\def\lifted{\mathrm{Lifted}}
\newcommand{\R}{\mathcal{R}}

\def\cov{\mathbb{C}\mathrm{ov}}
\def\rev{\mathrm{rev.}}
\def\MH{\mathrm{MH}}

\def\vara{\mathrm{var}}
\def\var{\mathbb{V}\mathrm{ar}}
\def\d{\mathrm{d}}

\def\esp{\mathbb{E}}
\newcommand{\pscal}[2]{\langle\,#1,#2\,\rangle}

\def\eps{\epsilon}

\def\bnd{\partial{\Xset}}
\def\interior{\mathring{\Xset}}

\def\D{\mathcal{E}}
\def\norm{\mathcal{N}}

\allowdisplaybreaks

\begin{document}

\def\figureautorefname{Figure}
\def\algorithmautorefname{Algorithm}
\def\sectionautorefname{Section}
\def\subsectionautorefname{Section}
\def\subsubsectionautorefname{Section}
\def\Propositionautorefname{Proposition}
\def\Theoremautorefname{Theorem}
\def\Lemmaautorefname{Lemma}
\def\Corollaryautorefname{Corollary}
\def\Exampleautorefname{Example}
\def\Remarkautorefname{Remark}
\def\Assumptionautorefname{Assumption}
\renewcommand*\footnoterule{}

\title{Theoretical guarantees for lifted samplers}

\author{Philippe Gagnon$^{1}$, Florian Maire$^{1}$}

\maketitle

\thispagestyle{empty}

\noindent $^{1}$Department of Mathematics and Statistics, Universit\'{e} de Montr\'{e}al.

\begin{abstract}
 Lifted samplers form a class of Markov chain Monte Carlo methods which has drawn a lot attention in recent years due to superior performance in challenging Bayesian applications. A canonical example of lifted samplers is the one that is derived from a random walk Metropolis algorithm for a totally-ordered state space such as the integers or the real numbers. The lifted sampler is derived by splitting into two the proposal distribution: one part in the increasing direction, and the other part in the decreasing direction. It keeps following a direction, until a rejection occurs, upon which it flips the direction. In terms of asymptotic variances, it outperforms the random walk Metropolis algorithm, regardless of the target distribution, at no additional computational cost. Other studies show, however, that beyond this simple case, lifted samplers do not always outperform their Metropolis counterparts. In this paper, we leverage the celebrated work of \cite{tierney1998note} to provide an analysis in a general framework encompassing a broad class of lifted samplers. Our finding is that, essentially, the asymptotic variances cannot increase by a factor of more than 2, regardless of the target distribution, the way the directions are induced, and the type of algorithm from which the lifted sampler is derived (be it a Metropolis--Hastings algorithm, a reversible jump algorithm, etc.). This result indicates that, while there is potentially a lot to gain from lifting a sampler, there is not much to lose.
\end{abstract}


\section{Introduction}\label{sec:intro}

In the context of Markov chain Monte Carlo (MCMC), a technique known as \textit{lifting} is applied to induce directions to follow by a Markov chain. It dates back at least to \cite{horowitz1991generalized}. In the 90's, it also appeared in \cite{gustafson1998guided} (yielding what was referred to as the \textit{guided walk}) and \cite{chen1999lifting}. It consists of a (simple) modification of a MCMC algorithm and an extension of a state-space $\Xset$, where each state $x \in \Xset$ is duplicated and combined with a \textit{direction} $\nu \in \{-1, +1\}$. The \textit{lifted} state-space is thus $\Xset \times \{-1, +1\}$ and the resulting class of sampling algorithms is referred to as \textit{lifted samplers}. A canonical example of a lifted sampler arises in the situation where the technique is applied to a random walk Metropolis--Hastings (MH, \cite{metropolis1953equation} and \cite{hastings1970monte}) algorithm and $\Xset = \{1, \ldots, n\}$, with $n$ being a positive integer. Let us say that the random walk MH algorithm is used to sample from $\pi$, referred to as the \textit{target distribution} in a sampling context, which is a probability distribution defined on a measurable space $(\Xset, \Xalg)$, where in this case $\Xset = \{1, \ldots, n\}$ corresponds to the support of $\pi$, $\Xalg$ being a sigma-algebra on $\Xset$. In this situation, the random walk MH algorithm first proceeds with a proposal $y = x + \nu$ with $\nu \sim \mathcal{U}\{-1, +1\}$ (except at the boundary, in which case a proposal outside of $\Xset$ is replaced by $x$), where $\mathcal{U}\{-1, +1\}$ denotes the uniform distribution over the set $\{-1, +1\}$. This is followed by an accept/reject step in which the proposal is accepted with probability
\[
 1 \wedge \frac{\pi(y)}{\pi(x)},
\]
where $a \wedge b = \min\{a, b\}$. If the proposal is accepted, the next state of the Markov chain is $y$. In the case of a rejection, the Markov chain remains at $x$. The resulting Markov chain is \emph{reversible} and exhibits a random walk behaviour. Such a behaviour may lead to significant auto-correlation in the Markov chain and an inefficiency of the MCMC estimators, often measured by the \textit{asymptotic variances}.

By inducing directions to follow, the goal of the lifting technique is to generate persistent movement and thus to lose this random walk behaviour, thereby decreasing the auto-correlation and gaining efficiency. In the canonical situation described above, applying the lifting technique to the random walk MH algorithm produces a lifted sampler which proceeds, when at a current state $(x, \nu) \in \Xset \times \{-1, +1\}$, with a \emph{deterministic} proposal $y = x + \nu$ (except at the boundary) which is accepted with the same probability as above, that is $1 \wedge \pi(y) / \pi(x)$. If the proposal is accepted, the next state of the Markov chain is $(y, \nu)$. In the case of a rejection (or when at the boundary), the next state of the Markov chain is $(x, -\nu)$. The Markov chain thus keeps following a direction, until a rejection occurs, upon which the direction is flipped. The lifting technique can be seen as a way to equip the resulting stochastic process with some memory of its past (the direction it comes from), while retaining the Markov property (on the extended state-space). The resulting Markov chain is \emph{non-reversible} and satisfies a \textit{skew detailed balance} (a formal definition is presented in the following), crucially ensuring that the product measure $\pi \otimes \mathcal{U}\{-1, +1\}$ is an invariant distribution. The lifted sampler is as easy to implement on a computer as its MH counterpart and it has the same computational complexity, as is often the case with lifted samplers. In a situation which is essentially the same as the canonical one described above, \cite{diaconis2000analysis} proved that lifting the random walk MH algorithm significantly decreases the mixing time, from quadratic to linear, when $\pi = \mathcal{U}\{1, \ldots, n\}$.

The lifting technique has drawn a lot attention in recent years due to success stories in challenging sampling problems (see, e.g., \cite{okabe2001replica}, \cite{sakai2016eigenvalue}, \cite{sakai2016irreversible}, \cite{vucelja2016lifting}, \cite{power2019accelerated}, \cite{herschlag2020non}, \cite{faizi2020efficient}, \cite{neal2020non}, \cite{gagnon2019NRJ}, \cite{andrieu2021peskun}, \cite{syed2019non}, \cite{kamatani2023non}, \cite{biron2024automatic} and \cite{ascolani2025fast}). For instance, in the context of Bayesian nested model selection, \cite{gagnon2019NRJ} proved that the lifted version of the reversible jump algorithm \citep{green1995reversible, green2003trans} is better. The result holds for any target distribution $\pi$, provided that the reversible jump algorithm with which the comparison is made explores the model space via a random walk. %In \cite{gagnon2019NRJ}, the reversible and non-reversible Markov chains described in the canonical situation above appear in similar forms, but for a model indicator in the reversible jump algorithm and the lifted version, respectively. The reversible and non-reversible Markov chains described in the canonical situation appear also in similar forms in other success stories mentioned above.
In \cite{andrieu2021peskun}, a general result is proved. It states that any lifted sampler is better than a specific reversible counterpart, the latter corresponding to the version of the lifted sampler in which the direction variable $\nu  \sim \mathcal{U}\{-1, +1\}$ is resampled at the beginning of each iteration. This specific reversible counterpart corresponds to the algorithm with which we wish to compare in the canonical situation described above and in \cite{gagnon2019NRJ}, but it is often not the case. A motivation of our work is to go beyond this specific situation and to provide a theoretical result when the algorithm with which we wish to compare is \emph{not} this specific reversible counterpart.

When going beyond this specific situation, one may encounter instances where lifted samplers are not better; see, e.g., Appendix A in \cite{gagnon2023} where the simulation of an Ising model is used to evaluate the performance of a MH algorithm and its lifted version. A partial order between the system states is exploited to induce directions to follow by the lifted sampler. It is shown that the lifted sampler outperforms its MH counterpart when the target distribution is not too concentrated and thus allows for persistent movement, whereas it is the opposite when the target distribution is too concentrated. The lifted sampler in \cite{gagnon2023} can be used to simulate any system formed from binary variables and thus finds application in Bayesian variable selection. % It corresponds to the discrete-time version of a specific sampler in \cite{power2019accelerated} and it is similar to samplers in  \cite{sakai2016eigenvalue} and \cite{faizi2020efficient}, all developed from another perspective than \cite{gagnon2023} (the notion of partial ordering is not identified nor exploited).
This work focuses on comparing MCMC algorithms for finite state-spaces through a weaker version of the Peskun ordering \citep{peskun1973optimum}.

With a counter-example like that in \cite{gagnon2023}, it becomes important to evaluate what can be the degree of inefficiency of lifted samplers relatively to the go-to reversible MCMC methods, given the versatility of the lifting technique and the appealing features of the resulting samplers. For instance, in the simulation of Ising models and in Bayesian variable selection, what can be the degree of inefficiency? And what about in Bayesian nested model selection when the point of comparison is the reversible jump algorithm which explores the model space via an \textit{asymmetric} random walk? An asymmetric random walk proposes to subtract or add 1, in this case to the model indicator variable, but with different probabilities. It appears frequently in the reversible jump algorithm (see, e.g., the multiple change-point example in \cite{green1995reversible}).

The main contribution of our work is to provide an answer to the aforementioned questions under arguably the most general framework. We proceed by leveraging the seminal work of \cite{tierney1998note} to define a lifted version of a generalized MH algorithm in the context of general state-spaces. Virtually any (reversible) MCMC method can be seen as a special case of this generalized MH algorithm (possibly up to minor adjustments), ranging from the traditional MH algorithm to the reversible jump algorithm. Other popular special cases include multiple-try Metropolis \citep{Liu2000}, pseudo-marginal \citep{beaumont2003estimation} and simulated tempering \citep{geyer1995annealing}. Our main theoretical result allows for a comparison between the generalized MH algorithm and its lifted version in terms of asymptotic variances. It essentially guarantees that the asymptotic variance of MCMC estimators produced by the lifted version cannot be more than twice that of estimators produced by the generalized MH algorithm. This result indicates that, \emph{while there is potentially a lot to gain from applying the lifting technique, there is not much to lose}. We also show that our result is optimal. The definition of the lifted version of the generalized MH algorithm allows to understand how a lifted sampler can be constructed under such a general framework, which adds a methodological contribution to our theoretical contribution.

The efficiency of MCMC algorithms is traditionally assessed by studying the characteristics of their Markov transition operators. To establish our main theoretical result, we need to connect the efficiency of two significantly different operators (those of the generalized MH and lifted algorithms): in addition to not be defined on the same domain, one is self-adjoint while the other is not, which further complicates the analysis. Another contribution of this work is to identify a specific auxiliary operator which acts a bridge and allows to connect the efficiency of the two aforementioned operators. This auxiliary operator is compared to the MH one through a Peskun-type ordering \citep{peskun1973optimum, tierney1998note} established via a careful analysis of the Markov kernels, yielding sharp bounds. A key ingredient of this analysis is \autoref{lemma:varphi} which is of independent interest as it states important properties of the class of functions $\varphi:(0, \infty) \rightarrow [0, 1]$ verifying $\varphi(r) = r \varphi(1 / r)$ for all $r \in (0, \infty)$ that appeared in several sampling contexts, for instance as efficient \textit{balancing functions} in \cite{zanella2019informed}. The connection between the generalized MH and lifted algorithms is completed by comparing the auxiliary operator with the lifted one using the result in \cite{andrieu2021peskun}.

We finish this section with a description of how the rest of the article is organized. In \autoref{sec:general}, we present the generalized MH algorithm, and define the generalized auxiliary and lifted samplers, each in a different sub-section, and our main theoretical result follows in the final sub-section. In \autoref{sec:examples}, we make the framework of \autoref{sec:general} more concrete by presenting the special case corresponding to the traditional MH framework. The traditional MH algorithm is explicit and simple, and the same is true for the auxiliary and lifted samplers. In \autoref{sec:examples}, within the concrete framework of that section, empirical results are also presented for an illustration of our main theoretical result. In particular, we revisit the Ising model example of \cite{gagnon2023}. As mentioned, in that paper, a methodology is introduced for applying the lifting technique in the context of partially-ordered finite state-spaces. Another contribution of this work is to highlight that this methodology is applicable in the general framework of \autoref{sec:general} whenever a partial order exists. The article finishes with a conclusion in \autoref{sec:conclusion}.

\section{MH and lifted algorithms: A general definition and comparison}\label{sec:general}

In \autoref{sec:general_MH}, we present the generalized MH algorithm, following \cite{tierney1998note}. Essentially any $\pi$-reversible MCMC algorithm which proceeds with a proposal distribution $Q(x, \cdot \,)$ and an accept/reject step can be seen as a special case of this generalized MH algorithm; $\pi$ and $Q(x, \cdot \,)$ thus do not necessarily admit densities with respect to a common dominating measure. From the generalized MH algorithm, we introduce the generalized auxiliary sampler in \autoref{sec:general_auxiliary}, which in turn is used to define the generalized lifted sampler in \autoref{sec:general_lifted}. In \autoref{sec:general_comparison}, we present the Peskun-type ordering, leading to our theoretical guarantees.

\subsection{Generalized MH algorithm}\label{sec:general_MH}

Let $\pi$, the target, be a probability distribution on $(\Xset, \Xalg)$ and $\{Q(x, \cdot\,): x \in \Xset\}$ a collection of (proposal) probability distributions on $(\Xset, \Xalg)$. % The support of $\pi$ and $Q(x, \cdot \,)$ is left unspecified.
Let $\mu$ and $\mu^T$ be two probability measures on the product space $(\Xset^2, \Xalg^{\otimes 2})$ defined as: $$
\mu(A\times B) = \int_A \pi(\d x) \, Q(x, B)\,,\qquad \mu^T(A\times B) = \int_B \pi(\d x) \, Q(x, A).
$$
Also, let $\eta := \mu + \mu^T$. From these definitions, it can be concluded that both $\mu$ and $\mu^T$ are absolutely continuous with respect to $\eta$ and that, by the Radon--Nikodym theorem, there exist two (essentially unique) measurable functions $h: \Xset^2 \to [0, \infty)$ and $h^T: \Xset^2 \to [0, \infty)$ such that
\[
 A \times B \in \Xalg^{\otimes 2}, \quad \mu(A \times B) = \int_{A \times B} h(x, y) \, \eta(\d x, \d y), \quad \mu^T(A \times B) = \int_{A \times B} h^T(x, y) \, \eta(\d x, \d y),
\]
with the usual notation $\d\mu / \d\eta := h$ and $\d\mu^T / \d\eta := h^T$.

The following proposition from \cite{tierney1998note} provides a characterization of the Radon--Nikodym derivative $\d\mu^T / \d\mu$ and an implicit form of it under the general framework of this section. This Radon--Nikodym derivative will then play a crucial role in the acceptance probability of the generalized MH Markov kernel, guaranteeing its $\pi$-reversibility.

\begin{Proposition}[\cite{tierney1998note}]\label{Prop:Tierney}
 \,
 \begin{enumerate}
  \item The function $h^T$ is such that $h^T(x, y) = h(y, x)$ for all $(x, y) \in \Xset^2$.
  \item The set
\[
 R := \{(x, y) \in \Xset^2: h(x, y) > 0 \quad \text{and} \quad h(y, x) > 0\},
\]
is symmetric.
 \item The probability measures $\mu$ and $\mu^T$ are mutually absolutely continuous on $R$ and mutually singular on the complement of $R$, $R^\mathsf{c}$.
 \item Let $r: R \rightarrow (0, \infty)$ be a function that verifies $r(x, y) = (\d\mu^T / \d\mu)(x, y)$ for all $(x, y) \in R$. Then,
   \begin{align}\label{eqn:r}
 r(x, y) = \frac{h(y, x)}{h(x, y)}, \quad \text{for all $(x, y) \in R$},
\end{align}
with $0 < r(x, y) < \infty$ and $r(x, y) = 1 / r(y, x)$ for all $(x, y) \in R$.
\end{enumerate}
\end{Proposition}

Let $P_{\MH}$ be the generalized MH kernel, defined for all $x \in \Xset$ as
\[
P_{\MH}(x, \d y) = Q(x, \d y) \, \alpha(x, y) + \delta_{x}(\d y) \int (1 - \alpha(x, u)) \, Q(x, \d u),
\]
where $\alpha: \Xset^2 \rightarrow [0,1]$ is the acceptance probability.

The following theorem from \cite{tierney1998note} states necessary conditions on the function $\alpha$ for $P_{\MH}$ to be $\pi$-reversible.

\begin{Theorem}[\cite{tierney1998note}]\label{Thm:Tierney}
The Markov kernel $P_{\MH}$ is $\pi$-reversible if and only if the following condition holds $\mu$-almost everywhere:
\begin{equation}\label{eqn:acc_Tierney}
\alpha(x, y) = \left\{
\begin{array}{ll}
r(x, y)  \, \alpha(y, x) & \text{if} \quad (x, y) \in R,  \cr
0 & \text{otherwise}.
\end{array}\right.
\end{equation}
\end{Theorem}

In this paper, we restrict our attention to the case where $\alpha$ is defined as follows:
\begin{equation}\label{eqn:acc_varphi}
\alpha(x, y) =\left\{
\begin{array}{ll}
  \varphi(r(x, y)) & \text{if $(x, y) \in R$},\\
   0 & \text{otherwise,}
\end{array}
\right.
\end{equation}
with $\varphi:(0, \infty) \rightarrow [0, 1]$ verifying $\varphi(r) = r \varphi(1 / r)$ for all $r \in (0, \infty)$. Such a function $\varphi$ allows to verify the condition in \autoref{Thm:Tierney}: for $(x, y) \in R$,
\[
 \alpha(x, y) = \varphi(r(x, y)) = r(x, y) \, \varphi(1 / r(x, y)) = r(x, y) \, \varphi(r(y, x)) = r(x, y) \, \alpha(y, x),
\]
by Result 4 of \autoref{Prop:Tierney}. The term $r(x, y)$ in \eqref{eqn:acc_varphi} can be seen as a \textit{summary statistic}, in the sense that the dependence on $(x, y)$ only appears through $r(x, y)$ in the acceptance probability. Nevertheless, it is sufficient, when combined with a function $\varphi$, given that the condition in \autoref{Thm:Tierney} is satisfied. The class of functions $\varphi:(0, \infty) \rightarrow [0, 1]$ verifying $\varphi(r) = r \varphi(1 / r)$ for all $r \in (0, \infty)$ allows for the usual acceptance probability $\varphi(r) = 1 \wedge r$ and Barker's acceptance probability $\varphi(r) = r / (1+ r)$ \citep{barker1965monte}. From \cite{tierney1998note}, we know that $\varphi(r) = 1 \wedge r$ yields the optimal acceptance probability (in the sense of minimizing the asymptotic variances).

This class of functions has interesting properties, as stated in \autoref{lemma:varphi}, which will be useful to prove our Peskun-type ordering in \autoref{sec:general_comparison}.

\begin{Lemma}
\label{lemma:varphi}
For any function $\varphi : (0,\infty)\to[0,1]$ verifying $\varphi(r) = r \varphi(1 / r)$ for all $r$, we have that:
\begin{enumerate}
  \item $\varphi$ is non-decreasing,
  \item $\varphi$ is continuous on $(0,\infty)$,
  \item $\varphi(r) \downarrow 0$ when $r \downarrow 0$,
  \item $\varphi$ is either a null function, that is $\varphi(r) = 0$ for all $r$, or it is a positive function,
  \item for any $a>0, b>0$, $\varphi(ab)\geq \varphi(a)\varphi(b)$,
  \item for any $a>0$, $b\in(0,1]$, $\varphi(ab)\geq b\varphi(a)$.
\end{enumerate}
\end{Lemma}

The proof of \autoref{lemma:varphi} is technical and thus deferred to \autoref{sec:proofs}.

What can be difficult in practice to implement a generalized MH algorithm is the identification of explicit forms for the function $r$ and the set $R$ (when available). A strategy to identify an explicit form for the function $r$ is to directly calculate the Radon--Nikodym derivative $\d\mu^T / \d\mu$, which represents a contribution of the articles introducing MCMC algorithms such as multiple-try Metropolis, pseudo-marginal, reversible jump and simulated tempering. Another strategy is to exploit the results in \cite{andrieu2020general}.

\subsection{Generalized auxiliary sampler}\label{sec:general_auxiliary}

We now introduce the generalized auxiliary sampler, constructed from the generalized MH algorithm presented in \autoref{sec:general_MH}. The generalized auxiliary sampler will in turn be used to construct the generalized lifted sampler in \autoref{sec:general_lifted}. The presentation of the generalized auxiliary and lifted samplers consists essentially of definitions; it will allow to understand how a lifted sampler can be constructed under the general framework of \autoref{sec:general}.

The generalized auxiliary sampler uses a notion of direction, but does not exploit it; it will be seen to correspond to the version of the generalized lifted sampler which resamples the direction variable $\nu  \sim \mathcal{U}\{-1, +1\}$ at the beginning of each iteration. As mentioned in \autoref{sec:intro}, it acts as a bridge between the generalized MH algorithm and the generalized lifted sampler, allowing a comparison between the MCMC estimators in terms of asymptotic variances.

For each $x \in \Xset$, we define two measurable sets $\neigh_{-1}(x)$ and $\neigh_{+1}(x)$ which will be referred to as \textit{directional neighbourhoods}. The collection of sets is assumed to satisfy two properties.
\begin{Assumption}\label{ass:sets}
\,
\begin{enumerate}

     \item[(i)] For any $x \in \Xset$, $Q(x, A) = Q(x, A \cap \neigh_{-1}(x)) + Q(x, A \cap \neigh_{+1}(x))$ for all $A \in \Xalg$.

    \item[(ii)] For any $(x, y) \in \Xset^2$ and $\nu \in \{-1, +1\}$, $y \in \neigh_\nu(x)$ \emph{if and only if} $x \in \neigh_{-\nu}(y)$, and for such a related couple $(x, y)$, $Q(x, \neigh_\nu(x)) > 0$ and $Q(y, \neigh_{-\nu}(y)) > 0$.
\end{enumerate}
\end{Assumption}

Let us consider the situation where $Q(x, \cdot \,)$ has a support given by, say, $\neigh(x)$; in this situation, the sets $\neigh_{-1}(x)$ and $\neigh_{+1}(x)$ can be constructed by splitting into two $\neigh(x)$ (this is what is done in our illustration in \autoref{sec:examples}). It is always possible to perform such a split when the state-space exhibits either a total or a partial order (because $\neigh(x)$ can be formed of states that can all be compared to $x$ via the total or partial order). In the examples in \autoref{sec:examples}, the state-space is either totally or partially ordered.

 We define a notion of \textit{interior} of $\Xset$ (that depends on the choice of directional neighbourhoods $\{\neigh_{-1}(x)\}$ and $\{\neigh_{+1}(x)\}$) as $\interior := \{x \in \Xset: Q(x, \neigh_{-1}(x)) Q(x, \neigh_{+1}(x)) > 0\}$. For instance, $Q(x, \neigh_{-1}(x)) Q(x, \neigh_{+1}(x)) = 0$ when one of the two directional neighbourhoods is empty. It is possible to have empty directional neighbourhoods when, for instance, $\Xset$, or a component of it, is a bounded subset of the integers (in particular, when the current state $x$ corresponds to the upper bound, we have that $\neigh_{-1}(x) = \neigh(x) = \{x - 1\}$ and $\neigh_{+1}(x) = \varnothing$, if $\neigh(x)$ comprises integers at a distance of 1 from $x$ in $\Xset$). Empty directional neighbourhoods thus arise in the context of Bayesian nested model selection discussed in \autoref{sec:intro}. They also arise in the Ising model example in \autoref{sec:examples} and in Bayesian variable selection.

 In the spirit of our notion of interior of $\Xset$, we define a notion of \textit{boundary} of $\Xset$: $\bnd := \Xset \backslash \interior$. For all $x \in \interior$ and $\nu \in \{-1, +1\}$, we denote the conditional proposal distribution given that the proposal belongs to $\neigh_{\nu}(x)$ as $Q_{\nu}(x, \cdot \,)$ which is such that
\begin{equation}\label{eqn:Q_nu}
 Q_\nu(x, A) = \frac{Q(x, A \cap \neigh_{\nu}(x))}{Q(x, \neigh_{\nu}(x))}, \quad A \in \Xalg.
\end{equation}

We define a probability distribution which will play the role of proposal distribution in the generalized auxiliary sampler. It essentially corresponds to a reweighing of the proposal distribution $Q$:
 \[
Q_{\rev}(x,\cdot\,) := \left\{
\begin{array}{ll}
(1/2) Q_{-1}(x, \cdot\,) + (1/2)Q_{+1}(x, \cdot\,) & \text{ if} \quad x\in \interior ,  \cr
(1/2)\delta_{x}+(1/2)Q(x,\cdot\,)& \text{ otherwise}.
\end{array}\right.
\]

We now present a lemma which is useful to prove that the Markov kernel of the generalized auxiliary sampler is $\pi$-reversible and which also allows to understand the implications of \autoref{ass:sets} on $Q$ and thus on $Q_{\rev}$.

\begin{Lemma}\label{lemma:P_rev}
 Suppose that \autoref{ass:sets} holds.
 \begin{enumerate}
  \item For all $x \in \Xset$, $Q(x, \{x\}) = 0$.
  \item For all $x \in \Xset$, $Q(x, \neigh_{-1}(x) \cap \neigh_{+1}(x)) = 0$.
  \item For all $x \in \Xset$, $Q(x, \neigh_{-1}(x) \cup \neigh_{+1}(x)) = 1$.
   \item When $x \in \bnd$, there exists $\nu \in \{-1, +1\}$ such that $Q(x, \neigh_\nu(x)) = 1$. Therefore, $Q_\nu$ can be defined as in \eqref{eqn:Q_nu} and in this case $Q_\nu(x, \cdot \,) = Q(x, \cdot \,)$.
 \end{enumerate}
\end{Lemma}

\begin{proof}
 Under \autoref{ass:sets}, we have that, if $x \in \neigh_\nu(x)$, then $x \in \neigh_{-\nu}(x)$, and from this we can deduce that $Q(x, \{x\}) = 2 Q(x, \{x\}) = 0$ for all $x \in \Xset$, concluding the proof of Result 1.

 We can also deduce that, for all $x \in \Xset$,
\[
 Q(x, \neigh_{-1}(x) \cap \neigh_{+1}(x)) = 2 Q(x, \neigh_{-1}(x) \cap \neigh_{+1}(x)) = 0,
\]
 and similarly,
\[
Q(x, \neigh_{-1}(x)^\mathsf{c} \cap \neigh_{+1}(x)^\mathsf{c}) = \sum_{\nu \in \{-1,+1\}} Q(x, \neigh_{\nu}(x)^\mathsf{c} \cap \neigh_{-\nu}(x)^\mathsf{c} \cap \neigh_{\nu}(x)) = 0,
\]
the latter implying that $Q(x, \neigh_{-1}(x) \cup \neigh_{+1}(x)) = 1$. This concludes the proof of Results 2 and 3.

Regarding Result 4, when $x \in \bnd$, there exists $\nu \in \{-1, +1\}$ such that $Q(x, \neigh_{-\nu}(x)) = 0$, which implies that $Q(x, \Xset) = Q(x, \Xset \cap \neigh_\nu(x)) = Q(x, \neigh_\nu(x)) = 1$ and that $Q(x, \cdot \,) = Q_\nu(x, \cdot \,)$.
\end{proof}

Similarly as in \autoref{sec:general_MH}, we define
$$
\mu_{\rev}(A\times B) := \int_A \pi(\d x) \, Q_{\rev}(x, B)\,,\qquad \mu_{\rev}^T(A\times B) := \int_B \pi(\d x) \, Q_{\rev}(x, A).
$$
We also define the following sets: $R_\nu := \{(x, y) \in \Xset^2: y \in \neigh_{\nu}(x)\}$ for $\nu \in \{-1, +1\}$, and $\Delta := \{(x, y) \in \Xset^2: x = y\}$.  Finally, supposing that \autoref{ass:sets} holds, we introduce a function $\beta: \Xset^2 \rightarrow [0, \infty)$:
\begin{equation}\label{eqn:beta}
\beta(x, y) = \left\{
\begin{array}{ll}
 0 & \text{ if} \quad (x, y) \in R^\mathsf{c}, \cr
 Q(x, \neigh_{\nu}(x)) / Q(y, \neigh_{-\nu}(y)) & \text{ if} \quad (x, y) \in R_{\nu} \cap R_{-\nu}^\mathsf{c} \cap \Delta^\mathsf{c} \cap R \text{ for some $\nu \in \{-1, +1\}$},   \cr
 1  & \text{ otherwise}.
\end{array}\right.
\end{equation}

We now present a result which will allow to define a valid acceptance probability in the generalized auxiliary sampler, in the sense of leading to a Markov kernel which is $\pi$-reversible. It highlights that the Radon--Nikodym derivative $\d\mu_{\rev}^T / \d\mu_{\rev}$ can be written as $\d\mu^T / \d\mu$ adjusted by the function $\beta$. It thus shows how to construct the acceptance probability in the generalized auxiliary and lifted samplers (as will be seen in \autoref{sec:general_lifted}) from that in the generalized MH algorithm.

\begin{Proposition}\label{prop:rev_derivative}
Suppose that \autoref{ass:sets} holds. The restriction of the measures $\mu_{\rev}$ and $\mu_{\rev}^T$ to $R$ are mutually absolutely continuous and, for all $(x, y) \in R$,
\[
 r_{\rev}(x, y) := \frac{\d\mu_{\rev}^T}{\d\mu_{\rev}}(x, y) = r(x, y) \, \beta(x, y),
\]
with $0 < r_{\rev}(x, y) < \infty$ and $r_{\rev}(x, y) = 1 / r_{\rev}(y, x)$.
\end{Proposition}

\noindent \textit{Sketch of the proof.} The proof of \autoref{prop:rev_derivative} is technical and thus deferred to \autoref{sec:proofs}. Here, we present a sketch of the proof of $ r_{\rev}(x, y) = (\d\mu_{\rev}^T / \d\mu_{\rev})(x, y) = r(x, y) \, \beta(x, y)$ for all $(x, y) \in R$. To simplify, let us consider the situation where $R = \Xset^2$ and $\interior = \Xset$.

We exploit a known result which follows from, e.g., Proposition 3 in \cite{roberts2004general}: two measures $\mu_1$ and $\mu_2$ are equal if and only if $\int \phi \, \d\mu_1 = \int \phi \, \d\mu_2$ for all bounded measurable functions $\phi$. The idea is thus to prove that, for any bounded measurable function $\phi: \Xset^2 \rightarrow \re$,
\[
 \int \phi \, \d\mu_{\rev}^T = \int \phi \, r_{\rev} \, \d\mu_{\rev}.
\]
Under the simplified situation, using the definition of $\mu_{\rev}^T$ we have that
\begin{align*}
\int \phi(x, y) \, \mu_{\rev}^T(\d x, \d y) &= \sum_{\nu \in \{-1, +1\}} \int  \phi(x, y) \, \pi(\d y) \, \frac{1}{2} Q_{\nu}(y, \d x).
\end{align*}
We analyse each term of the sum separately:
\begin{align*}
\int \phi(x, y) \, \pi(\d y) \, \frac{1}{2} Q_{\nu}(y, \d x) &= \int\frac{\phi(x, y) \, \ind(x \in \neigh_{\nu}(y))}{2 Q(y, \neigh_{\nu}(y))}  \, \pi(\d y) \, Q(y, \d x) \cr
&= \int \frac{\phi(x, y)\,  \ind(x \in \neigh_{\nu}(y))}{2 Q(y, \neigh_{\nu}(y))}  \, \mu^T(\d x, \d y) \cr
&= \int \frac{\phi(x, y) \,\ind(y \in \neigh_{-\nu}(x))}{2 Q(y, \neigh_{\nu}(y))} \, r(x, y)  \, \mu(\d x, \d y) \cr
&= \int \phi(x, y) \, r(x, y) \, \frac{Q(x, \neigh_{-\nu}(x))}{Q(y, \neigh_{\nu}(y))}  \, \pi(\d x) \, \frac{1}{2} Q_{-\nu}(x, \d y),
\end{align*}
using the definitions of $Q_\nu$ and $\mu^T$ in the first and second equalities, respectively, and that $(y, x) \in R_{\nu}$ if and only if $(x, y) \in R_{-\nu}$ and \autoref{Prop:Tierney} in the third equality, by combining $\phi(x, y)$ and $\ind(x \in \neigh_{\nu}(y))/ (2 Q(y, \neigh_{\nu}(y)))$ into a new bounded measurable function. Note that $1 / Q(y, \neigh_{\nu}(y))$ is not necessarily bounded, but we can limit ourselves to the case where $Q(y, \neigh_{\nu}(y))$ is bounded from below by an arbitrarily small constant and use a limiting argument with the monotone convergence theorem.

We obtain $\int \phi \, r_{\rev} \, \d\mu_{\rev}$ by summing over $\nu$ (and after verifying technical details), recalling the definition of $r_{\rev}$ (see the statement of the proposition) and $\beta$ in \eqref{eqn:beta}. \hfill \qedsymbol

Following \cite{tierney1998note}, \autoref{prop:rev_derivative} allows to state that the Markov kernel of the generalized auxiliary sampler defined as
\[
 P_{\rev}(x, \d y) := Q_{\rev}(x, \d y) \, \alpha_{\rev}(x, y) + \delta_{x}(\d y) \int (1 - \alpha_{\rev}(x, u)) \, Q_{\rev}(x, \d u), \quad \text{for all $x \in \Xset$},
\]
with the function $\alpha_{\rev}$ given by
\[
(x, y) \in \Xset^2, \quad \alpha_{\rev}(x, y) = \left\{
\begin{array}{ll}
  \varphi(r_{\rev}(x, y)) & \text{if $(x, y) \in R$},\\
   0 & \text{otherwise,}
\end{array}
\right.
\]
is in fact a generalized MH kernel, but one that uses $Q_{\rev}$ as proposal kernel. It is thus $\pi$-reversible by \autoref{Thm:Tierney}. \autoref{lemma:P_rev} indicates that, under \autoref{ass:sets}, for $\mu_{\rev}$-almost all $(x, y)$ with $(x, y) \in R$, the acceptance probability $\alpha_{\rev}(x, y)$ is equal to
\[
 \varphi(r_{\rev}(x, y)) = \varphi(r(x, y) \, \beta(x, y)) = \varphi\left(r(x, y) \, \frac{Q(x, \neigh_{\nu}(x))}{Q(y, \neigh_{-\nu}(y))}\right)
\]
for some $\nu \in \{-1, +1\}$, except if $y = x$, in which case $\varphi(r_{\rev}(x, y)) = 1$.

\subsection{Generalized lifted samplers}\label{sec:general_lifted}

From the generalized auxiliary sampler introduced in \autoref{sec:general_auxiliary}, we define the generalized lifted sampler. For all $x \in \Xset$ and $\nu \in \{-1, +1\}$, let $T_\nu(x, \cdot \,)$ be a sub-probability measure on $(\Xset, \Xalg)$ defined as:
\[
A \in \Xalg, \quad T_\nu(x, A) = \left\{
\begin{array}{ll}
  \int_A \alpha_{\rev}(x, y) \, Q_\nu(x, \d y) & \text{if $x \in \interior$ or $x \in \bnd$ with $Q(x, \neigh_\nu(x)) = 1$},\\
   0 & \text{otherwise.}
\end{array}
\right.
\]
Let $P_{\lifted}$ be a Markov kernel defined on the state-space $\Xset \times \{-1, +1\}$ as: for all $x \in \Xset$ and $\nu \in \{-1, +1\}$,
\[
 P_{\lifted}((x, \nu), \d(y, \nu')) = T_{\nu}(x, \d y) \, \delta_\nu(\d\nu') + \delta_{x}(\d y) \, \delta_{-\nu}(\d\nu') \, (1 - T_{\nu}(x, \Xset)).
\]

Under \autoref{ass:sets}, if $(x, \nu)$ is such that $x \in \interior$ or $x \in \bnd$ with $Q(x, \neigh_\nu(x)) = 1$, the lifted sampler proceeds with a proposal $y \sim Q_\nu(x, \cdot \,)$ such that, with probability 1,
\[
 \alpha_{\rev}(x, y) = \varphi(r_{\rev}(x, y)) = \varphi\left(r(x, y) \, \frac{Q(x, \neigh_{\nu}(x))}{Q(y, \neigh_{-\nu}(y))}\right),
\]
 if $(x, y) \in R$ or 0 if $(x, y) \notin R$, recalling the definition of $\beta$ in \eqref{eqn:beta}. Indeed, under \autoref{ass:sets}, if $(x, \nu)$ is such that $x \in \interior$ or $x \in \bnd$ with $Q(x, \neigh_\nu(x)) = 1$, we have that $(x, y) \in R_{\nu} \cap R_{-\nu}^\mathsf{c} \cap \Delta^\mathsf{c}$ for $Q_\nu(x, \cdot \,)$-almost all $y$ because
 \[
 Q_\nu(x, \neigh_\nu(x) \cap \neigh_{-\nu}(x)^\mathsf{c} \cap \{x\}^\mathsf{c}) = 1.
\]
 The latter follows from
\[
 Q_\nu(x, \neigh_\nu(x)^\mathsf{c} \cup \neigh_{-\nu}(x) \cup \{x\}) \leq \frac{Q(x, \neigh_\nu(x)^\mathsf{c} \cap \neigh_\nu(x))}{Q(x, \neigh_\nu(x))} + \frac{Q(x, \neigh_{-\nu}(x) \cap \neigh_\nu(x))}{Q(x, \neigh_\nu(x))} + \frac{Q(x, \{x\} \cap \neigh_\nu(x))}{Q(x, \neigh_\nu(x))} = 0,
\]
by the union bound and \autoref{lemma:P_rev}.

The following proposition states that $P_{\lifted}$ satisfies a skewed detailed balance \citep{andrieu2021peskun}, implying that the product measure $\pi \otimes \mathcal{U}\{-1, +1\}$ is an invariant distribution. Following \cite{andrieu2021peskun}, $P_{\lifted}$ is said to satisfy a skewed detailed balance if the part of the kernel associated to a move of the $X$ component, that is $T_{\nu}(x, \d y)$, satisfies
\[
 \pi(\d x) \, T_{+1}(x, \d y) = \pi(\d y) \, T_{-1}(y, \d x).
\]
More formally, it corresponds to proving that the two following measures on the product space $(\Xset^2,\Xalg^{\otimes 2})$ are equal:
$$
\mu_{+1}(A\times B) = \int_A \pi(\d x) \, T_{+1}(x, B) \quad \text{and} \quad \mu_{-1}(A\times B) = \int_B \pi(\d x) \, T_{-1}(x, A).
$$
For completeness, the proof that a skewed detailed balance implies that the product measure $\pi \otimes \mathcal{U}\{-1, +1\}$ is an invariant distribution is provided in \autoref{sec:proofs}.

\begin{Proposition}\label{prop:generalized_lifted}
 Supposing that \autoref{ass:sets} holds, $P_{\lifted}$ satisfies a skewed detailed balance.
\end{Proposition}

\noindent \textit{Sketch of the proof.} As for \autoref{prop:rev_derivative}, here we present a sketch of the proof of \autoref{prop:generalized_lifted} as the formal proof is technical and thus deferred to \autoref{sec:proofs}. As in the sketch of the proof of \autoref{prop:rev_derivative}, let us consider the simplified situation where $R = \Xset^2$ and $\interior = \Xset$.

The idea is to prove that, for any bounded measurable function $\phi: \Xset^2 \rightarrow \re$,
\[
 \int \phi \, \d\mu_{+1} = \int \phi \, \d\mu_{-1}.
\]
Under the simplified situation and using the definition of $\mu_{+1}$, we have that
\[
\mu_{+1}(\d x, \d y) = \alpha_{\rev}(x, y) \, \pi(\d x) \, Q_{+1}(x, \d y).
\]
The measure $\pi(\d x) \, Q_{+1}(x, \d y)$ is only positive on the set where $(x, y)$ is such that $y \in \neigh_{+1}(x)$. Thus, we can limit ourselves to this set, which corresponds to $R_{+1}$. Also, we can add $Q_{-1}(x, \d y)$ because it is null on this set and obtain
\begin{align*}
 \int \phi \, \d\mu_{+1} &= \int_{R_{+1}} \phi(x, y) \, \alpha_{\rev}(x, y) \, \pi(\d x) \, (Q_{+1}(x, \d y) + Q_{-1}(x, \d y)) \cr
 &= \int_{R_{+1}}2 \, \phi(x, y) \, \varphi(r_{\rev}(x, y)) \, \mu_{\rev}(\d x, \d y) \cr
 &= \int_{R_{+1}}2 \, \phi(x, y) \, \varphi(1 / r_{\rev}(x, y)) \, r_{\rev}(x, y) \, \mu_{\rev}(\d x, \d y) \cr
  &= \int_{R_{+1}}2 \, \phi(x, y) \, \varphi(1 / r_{\rev}(x, y)) \, \mu_{\rev}^T(\d x, \d y) \cr
  &= \int_{R_{+1}} \phi(x, y) \, \varphi(r_{\rev}(y, x)) \, \pi(\d y) \, (Q_{+1}(y, \d x) + Q_{-1}(y, \d x)),
\end{align*}
using the definitions of $\alpha_{\rev}$ and $\mu_{\rev}$ in the second equality (see \autoref{sec:general_auxiliary}), the property of the function $\varphi$ in the third equality, \autoref{prop:rev_derivative} in the fourth equality (recalling that $\varphi$ is bounded) and the definition of $\mu_{\rev}^T$ in the last equality. From there, we proceed symmetrically as above to obtain the result, noting that $(x, y) \in R_{+1}$ if and only if $(y, x) \in R_{-1}$.
 \hfill \qedsymbol

 Note that the definition of lifted sampler presented in this sub-section is in fact not the most general. \cite{andrieu2021peskun} present a more general definition in which $Q_{-1}$ and $Q_{+1}$ are not necessarily defined through a decomposition of $Q$; they can be any proposal kernels. In such generality, it can be shown that it is not possible to establish a comparison with the MH algorithm (see \autoref{sec:noPeskun} for a demonstration). At least, it is not possible using the route proposed in \autoref{sec:general_comparison} that allows for a comparison when considering the definition of $Q_{-1}$ and $Q_{+1}$ in \autoref{sec:general_auxiliary}. Note that most applications of the lifting technique in practice use this definition.

\subsection{The comparison}\label{sec:general_comparison}

Equipped with the definitions and results of Sections \ref{sec:general_MH}-\ref{sec:general_lifted}, we present in this section our theoretical guarantees for the generalized lifted sampler. They take the form of an upper bound on the asymptotic variances of the associated MCMC estimators, where the bound is a function of the asymptotic variances of the MCMC estimators associated with the generalized MH algorithm.

Let us recall that, for an homogeneous Markov chain $\{X_k\}$ of transition kernel $P$ leaving $\pi$ invariant started in stationarity, the asymptotic variance of a MCMC estimator of $\pi f$, the expectation of $f(X)$ under $X\sim\pi$, is defined as
\begin{align*}
 \vara(f, P) := \lim_{T \rightarrow \infty} T \var\left[\frac{1}{T} \sum_{k=1}^T f(X_k)\right],
\end{align*}
whenever the limit exists. When $P$ is reversible and $\pi(f^2) < \infty$, the limit always exists, but may be infinite \citep{tierney1998note}. When $P$ is additionally ergodic and the limit is finite,
\begin{align}\label{eqn:asymp_var}
 \vara(f, P) =  \var[f(X_0)] + 2 \sum_{k = 1}^\infty \cov[f(X_0), f(X_k)],
 \end{align}
 by Theorem 4 of \cite{10.1214/ECP.v12-1336}.

A standard route to compare the asymptotic variances of MCMC estimators produced by two different samplers (e.g., $P_{\lifted}$ and $P_{\MH}$) is by establishing a Peskun-type ordering between the Markov kernels. In our context, a challenge is that $P_{\lifted}$ and $P_{\MH}$ are not defined on the same state-space and thus cannot be ordered. As mentioned in \autoref{sec:intro}, Theorem 7 in \cite{andrieu2021peskun} allows to compare $P_{\lifted}$ with a $\pi$-reversible kernel corresponding to the version of $P_{\lifted}$ which resamples the direction variable at the beginning of each iteration. In our analysis, $P_{\rev}$ plays the role of this $\pi$-reversible kernel. In Theorem 7 in \cite{andrieu2021peskun}, $P_{\rev}$ is seen as a Markov kernel that can also be defined on the extended state-space $\Xset \times \{-1, +1\}$ and only estimators of expectations of functions of $X \sim \pi$ are considered. When focusing on functions $f$ of solely the first argument, extending the state-space to include the direction variable in $P_{\rev}$ is superfluous. Let us be more specific and consider such a function $f: \Xset \times \{-1, +1\} \rightarrow \re$ with $f(x, \nu) = g(x)$. For such a function $f$, the asymptotic variance of an estimator produced by $P_{\rev}$, but when the latter is defined on the extended state-space, is equivalent to the asymptotic variance for $g$ of an estimator produced by $P_{\rev}$ defined on $\Xset$ (as in \autoref{sec:general_auxiliary}). The kernel $P_{\rev}$ is thus used as an intermediate kernel through which comparison of the asymptotic variances of the MCMC estimators associated with $P_{\MH}$ and $P_{\lifted}$ is possible. With an abuse of notation, we will use $f$ to denote both a function acting on $\Xset \times \{-1, +1\}$ of solely the first argument and a function acting on $\Xset$. Only such functions acting on $\Xset \times \{-1, +1\}$ of solely the first argument will be considered in the following.

The proof of our theoretical result is thus based on establishing a Peskun-type ordering between $P_{\MH}$ and $P_{\rev}$. The fact that the MCMC estimators produced by $P_{\rev}$ cannot have smaller asymptotic variances than those produced by $P_{\lifted}$ \citep[Theorem 7]{andrieu2021peskun} concludes the argument. Intuitively, this latter fact can be seen as following from the observation that $P_{\rev}$ proceeds as $P_{\lifted}$ with only \emph{one} difference: $P_{\rev}$ resamples the direction variable at the beginning of each iteration instead of keeping it fixed until a rejection occurs. The Markov kernel $P_{\rev}$ is thus seen as a reversible counterpart to $P_{\lifted}$.

In our theoretical result, we use the notion of $\lambda$-asymptotic variance (as in \cite{andrieu2021peskun}) defined, for an homogeneous Markov chain $\{X_k\}$ of transition kernel $P$ leaving $\pi$ invariant started in stationarity, as
\[
 \vara_\lambda(f, P) := \var[f(X_0)] + 2 \sum_{k = 1}^\infty \lambda^k \, \cov[f(X_0), f(X_k)], \quad \lambda \in [0, 1),
\]
whenever $\pi(f^2) < \infty$. Note that $\lim_{\lambda \rightarrow 1} \vara_\lambda(f, P) = \vara(f, P)$ under some conditions (not requiring reversibility). A general condition is that $\sum_{k = 1}^\infty |\cov[f(X_0), f(X_k)]| < \infty$, which can be established under fairly mild assumptions \citep[Corollary 3]{andrieu2021peskun}. A more concrete (but stronger) condition is uniform ergodicity \citep{gagnon2023}.

We are now ready to present our result.

\begin{Theorem}\label{thm2}
 Suppose that \autoref{ass:sets} holds. For any $x \in \Xset$ and $A \in \Xalg$, we have that
 \[
  P_{\rev}(x, A \setminus \{x\}) \geq \frac{1}{2} P_{\MH}(x, A \setminus \{x\}).
 \]
 Therefore, for any $f$ such that $\pi(f^2) < \infty$ and whenever $\lim_{\lambda \rightarrow 1} \vara_\lambda(f, P_{\lifted}) = \vara(f, P_{\lifted})$,
 \[
  \vara(f, P_{\lifted}) \leq \vara(f, P_{\rev}) \leq 2\vara(f, P_{\MH}) + \var[f(X)], \quad X \sim \pi.
 \]
\end{Theorem}

\begin{proof} With some abuse of notation, let, for $x \in \Xset$, $R_{x} := \{y \in \Xset: (x, y) \in R\}$. Let us first consider the case where $x \in \interior$. For $A \in \Xalg$, we have  that
 \begin{align}
  P_{\rev}(x, A \setminus \{x\}) = \int_{A \setminus \{x\}} Q_{\rev}(x, \d y) \, \alpha_{\rev}(x, y)
   &=\frac{1}{2} \sum_{\nu \in\{-1, +1\}} \int_{A \setminus \{x\}} Q_{\nu}(x, \d y) \, \alpha_{\rev}(x, y) \nonumber \\
%   &=\frac{1}{2} \sum_{\nu \in\{-1, +1\}} \int_{A \cap R_{x} \cap \neigh_{\nu}(x) \cap \neigh_{-\nu}(x)^\mathsf{c} \setminus \{x\}} Q_{\nu}(x, \d y) \, \alpha_{\rev}(x, y) \nonumber \\
  &= \frac{1}{2} \sum_{\nu \in\{-1, +1\}} \int_{A \cap R_{x} \cap \neigh_{\nu}(x) \cap \neigh_{-\nu}(x)^\mathsf{c} \setminus \{x\}} Q(x, \d y) \, \frac{\alpha_{\rev}(x, y)}{Q(x, \neigh_{\nu}(x))}, \label{eqn:proof1}
 \end{align}
 using the definitions of $P_{\rev}$ and $Q_{\rev}(x, \cdot \,)$, that, for any $x \in \interior$, $Q_{\nu}(x, \cdot \,)$ assigns null measure to $(\neigh_{\nu}(x) \cap \neigh_{-\nu}(x)^\mathsf{c})^\mathsf{c} = \neigh_{\nu}(x)^\mathsf{c} \cup \neigh_{-\nu}(x)$ by definition and \autoref{lemma:P_rev}, and that $\alpha_{\rev}(x, y) = 0$ for all $(x, y) \notin R$.

 For all $(x, y)$ with $x \in \interior$ and $y \in R_{x} \cap \neigh_{\nu}(x) \cap \neigh_{-\nu}(x)^\mathsf{c} \setminus \{x\}$,
 \begin{align}
  \frac{\alpha_{\rev}(x, y)}{Q(x, \neigh_{\nu}(x))} = \frac{\varphi(r(x, y) \beta(x, y))}{Q(x, \neigh_{\nu}(x))} &\geq \frac{\varphi(r(x, y) (1 \wedge \beta(x, y)))}{Q(x, \neigh_{\nu}(x))} \nonumber \\
  &\geq \frac{\varphi(r(x, y)) (1 \wedge \beta(x, y))}{Q(x, \neigh_{\nu}(x))} \nonumber \\
  &= \varphi(r(x, y)) \left(\frac{1}{Q(x, \neigh_{\nu}(x))} \wedge \frac{1}{Q(y, \neigh_{-\nu}(y))}\right) \label{eqn:proof2} \\
  &\geq \varphi(r(x, y)) = \alpha(x, y), \nonumber
 \end{align}
 using the definitions of $\alpha_{\rev}$ and $\beta$ (see \autoref{sec:general_auxiliary}), and Results 1 and 6 of \autoref{lemma:varphi}.

 Therefore,
  \begin{align*}
  P_{\rev}(x, A \setminus \{x\}) &\geq \frac{1}{2}\sum_{\nu \in\{-1, +1\}} \int_{A \cap R_{x} \cap \neigh_{\nu}(x) \cap \neigh_{-\nu}(x)^\mathsf{c} \setminus \{x\}} Q(x, \d y) \, \alpha(x, y) \cr
  &= \frac{1}{2} \int_{A \cap R_{x} \setminus \{x\}} \left(\ind(y \in \neigh_{-1}(x) \cap \neigh_{+1}(x)^\mathsf{c}) + \ind(y \in \neigh_{+1}(x) \cap \neigh_{-1}(x)^\mathsf{c})\right) Q(x, \d y) \, \alpha(x, y) \cr
  &= \frac{1}{2} \int_{A \setminus \{x\}} Q(x, \d y) \, \alpha(x, y) = \frac{1}{2} P_{\MH}(x, A \setminus \{x\}),
 \end{align*}
 using in the first equality that, for any $x \in \Xset$,
 \[
  \ind(y \in \neigh_{-1}(x) \cap \neigh_{+1}(x)^\mathsf{c}) + \ind(y \in \neigh_{+1}(x) \cap \neigh_{-1}(x)^\mathsf{c}) = 1 \quad \text{for \quad $Q(x, \cdot \,)$-almost all $y$},
 \]
 which can be deduced from \autoref{lemma:P_rev} (the formal proof is similar to that of \autoref{lemma3} in \autoref{sec:proofs}), and in the second equality that $\alpha(x, y) = 0$ for all $(x, y) \notin R$.

 The case $x \notin \interior$ is done similarly after noticing that, in this case, there exists $\nu \in \{-1, +1\}$ such that $Q(x, A) = Q_\nu(x, A)$ and $Q(x, \neigh_{\nu}(x)) = 1$ by \autoref{lemma:P_rev}.

   The result
 \[
  P_{\rev}(x, A \setminus \{x\}) \geq \frac{1}{2} P_{\MH}(x, A \setminus \{x\}), \quad \text{for all $x \in \Xset$ and $A \in \Xalg$},
 \]
 implies that, for any $f$ such that $\pi(f^2) < \infty$,
  \[
   \vara(f, P_{\rev}) \leq 2\vara(f, P_{\MH}) + \var[f(X)], \quad X \sim \pi,
 \]
 by \autoref{Lemma:ordering}.

 The proof of
 \begin{equation}\label{eqn:proof3}
  \vara_\lambda(f, P_{\lifted}) \leq \vara_\lambda(f, P_{\rev}),
 \end{equation}
 in the case where $\varphi(r) = 1 \wedge r$ follows from Theorem 7 in \cite{andrieu2021peskun}. It can be readily verified that the result holds for any function $\varphi$ considered here using the same proof technique as in \cite{andrieu2021peskun}. The proof is concluded by taking the limit $\lambda \rightarrow 1$ on both side of the inequality in \eqref{eqn:proof3}.
 \end{proof}

 We now make a few remarks about \autoref{thm2}. Firstly, the result reflects a notion of \textit{universality}: in addition to being stated in a context of general state-spaces, it is valid for \emph{any} target distribution $\pi$, \emph{any} proposal mechanism (defining the generalized MH algorithm from which $P_{\rev}$ and $P_{\lifted}$ are constructed, be it the usual MH, multiple-try, pseudo-marginal or reversible jump), and for \emph{any} way of inducing directions (through the definition of $\{\neigh_{-1}(x)\}$ and $\{\neigh_{+1}(x)\}$). Also, it holds for any function $f$ belonging to a class of functions (the functions $f$ such that $\pi(f^2) < \infty$ and $\lim_{\lambda \rightarrow 1} \vara_\lambda(f, P_{\lifted}) = \vara(f, P_{\lifted})$). In different contexts, performance comparison has been established through an order between \textit{average} asymptotic variances \citep{chen2012optimal} and \textit{maximal} asymptotic variances \citep{frigessi1992optimal}; such an ordering is thus implied by results such as \autoref{thm2}.

Secondly, we observe in the proof that the factor
\[
 \left(\frac{1}{Q(x, \neigh_{\nu}(x))} \wedge \frac{1}{Q(y, \neigh_{-\nu}(y))}\right)
\]
in \eqref{eqn:proof2} is an explanation of the difference between $P_{\rev}$ and $P_{\MH}$. It is a measure of whether the directional neighbourhoods $\neigh_{-1}(x)$ and $\neigh_{+1}(x)$ (and $\neigh_{-1}(y)$ and $\neigh_{+1}(y)$) are well balanced or not. When $Q(x, \neigh_{-1}(x)) = Q(x, \neigh_{+1}(x)) = 1 / 2$ for all $x \in \Xset$, we say that we are in a situation of \textit{perfect balance}. In this situation, that factor is always equal to 1 and $P_{\rev}(x, A \setminus \{x\}) \geq P_{\MH}(x, A \setminus \{x\})$ for all $x \in \Xset$ and $A \in \Xalg$. We thus have a guarantee that the lifted sampler outperforms the MH one, that is $\vara(f, P_{\lifted}) \leq \vara(f, P_{\MH})$ whenever $\lim_{\lambda \rightarrow 1} \vara_\lambda(f, P_{\lifted}) = \vara(f, P_{\lifted})$. It is for instance the case with the guided walk, by symmetry of the proposal distribution in random walk MH. Theorem 7 in \cite{andrieu2021peskun} is stated under the general framework of this section (where $\pi$ and $Q(x, \cdot \,)$  do not necessarily admit densities with respect to a common dominating measure), but with the restriction to the situation of perfect balance and to the function $\varphi(r) = 1 \wedge r$. Note that before introducing the inequalities in the proof, we had that $P_{\rev}(x, A \setminus \{x\}) = P_{\MH}(x, A \setminus \{x\})$ under the condition $Q(x, \neigh_{-1}(x)) = Q(x, \neigh_{+1}(x)) = 1 / 2$  for all $x \in \Xset$ (see \eqref{eqn:proof1}). This indicates that the inequalities in the proof are tight because, under this condition, the inequalities yield $P_{\rev}(x, A \setminus \{x\}) \geq P_{\MH}(x, A \setminus \{x\})$ but we know that $P_{\rev}(x, A \setminus \{x\}) = P_{\MH}(x, A \setminus \{x\})$.

Thirdly, in a context that fits within the framework of this section, \cite{gagnon2023} prove that $P_{\rev}(x, A \setminus \{x\}) / P_{\MH}(x, A \setminus \{x\}) \rightarrow 1/2$ for some $x$ and $A$, as a dimension parameter grows without bounds (Proposition 2 in that paper). We can therefore conclude that our ordering between $P_{\rev}$ and $P_{\MH}$ is (essentially) optimal, in the sense that it is (essentially) not possible to obtain a better ordering without additional assumptions. In \autoref{sec:optimality}, we show that our bound on the asymptotic variances is (essentially) optimal. \autoref{thm2} is thus (essentially) optimal.

Fourthly, in a context of comparison between two samplers, the magnitude of the asymptotic variances, which is proportional to $\var[f(X)]$ (see \eqref{eqn:asymp_var}), is irrelevant. Thus, it is appropriate to focus on standardized functions. When $\var[f(X)] = 1$, the expression in \eqref{eqn:asymp_var} corresponds to the integrated autocorrelation time. In this case, the factor $2$ and the additive term in the bound in \autoref{thm2} on the asymptotic variances (and thus on the integrated autocorrelation times) are independent of all problem parameters, whether it is the dimension of the state-space, the sample size in Bayesian statistics contexts, etc. This implies that the generalized lifted sampler is \emph{at worst} comparable to the generalized MH algorithm, in the sense that the asymptotic variances can be larger, but at worst, the factor $2$ and the additive term explaining the difference between the asymptotic variances do not deteriorate when a problem parameter changes. All that suggests the following practical suggestion: \textit{when one has a way of inducing directions which does not significantly increase the computational cost of the algorithm, better do it as there is not much to lose by lifting a sampler, but there is potentially a lot to gain (as observed many times in the past)}.

Finally, the Peskun-type ordering between $P_{\rev}$ and $P_{\MH}$ in \autoref{thm2} also allows to obtain an order between the spectral gaps \citep[Lemma 33]{andrieu2018uniform}: the spectral gap of $P_{\rev}$ is greater than or equal to half of that of $P_{\MH}$. This result leads to an order between the mixing times in total variation, provided that $P_{\MH}$ has a spectral gap (see, e.g., \cite{andrieu2018hypocoercivity}). Therefore, if one has established an order between the mixing times of $P_{\lifted}$ and $P_{\rev}$, then an order between the mixing times of $P_{\lifted}$ and $P_{\MH}$ follows. Lifted samplers do not dominate their reversible counterparts in terms of mixing times because the former can exhibit quasi-periodic or even periodic behaviour (see, e.g., \cite{vialaret2020convergence}).

\section{Illustration: The case of common dominating measure}\label{sec:examples}

In this section, we make more concrete the framework of \autoref{sec:general} by considering the case where $\pi$ and $Q(x, \cdot \,)$ admit densities with respect to a common dominating measure. In this case, $P_{\MH}$ corresponds to the usual MH kernel. It is simple and explicit, which allows simple and explicit forms of $P_{\rev}$ and $P_{\lifted}$. All these kernels are presented in \autoref{sec:samplers_dominating}. Next, we illustrate \autoref{thm2} by presenting empirical results and instances of lifted samplers in two concrete situations fitting within the framework of this section: one where the state-space is finite and partially ordered (\autoref{sec:example1}), and one where the state-space is the real numbers, and thus, totally ordered (\autoref{sec:example2}). We use this illustration to highlight that, even in situations where it is suspected that the lifted sampler will not offer a great performance, the asymptotic variances are at worst comparable to those of MH, as guaranteed by \autoref{thm2}. In particular, we show that, in very unfavourable situations, the upper bound in \autoref{thm2} is attained. The code to reproduce our numerical results is available online (see ancillary files on \url{https://arxiv.org/abs/2405.15952}).

\subsection{MH, auxiliary and lifted samplers}\label{sec:samplers_dominating}

As mentioned, we consider the case where $\pi$ and $Q(x, \cdot \,)$ admit densities with respect to a common dominating measure, denoted by $\d y$ for simplicity. With an abuse of notation, we also denote these densities by $\pi$ and $Q(x, \cdot \,)$. To simplify, we consider that $\Xset \subset \re^d$ is the support of $\pi$, with $d$ a positive integer, and that, for all $x \in \Xset$, $\neigh(x) \subset \Xset$ is the support of $Q(x, \cdot \,)$, where $\{\neigh(x)\}$ yields a neighbourhood structure on $\Xset$. We also consider that $Q(x, \{x\}) = 0$ for all $x \in \Xset$ (i.e., the probability is equal to 0, but not necessarily the density). Finally, we consider that: for any $(x, y) \in \Xset^2$,
\[
 y \in \neigh(x) \textit{ if and only if } x \in \neigh(y).
\]

Under the framework just described, the conclusion of \autoref{Thm:Tierney} holds, meaning that $P_{\MH}$ is $\pi$-reversible, in the case where $P_{\MH}$ using the usual acceptance probability is defined for all $x \in \Xset$ as
\[
P_{\MH}(x, \d y) = Q(x, \d y) \, \alpha(x, y) + \delta_{x}(\d y) \int (1 - \alpha(x, u)) \, Q(x, \d u),
\]
with, for all $(x, y) \in \Xset^2$, $\alpha(x, y) = 0$ whenever $y \notin \neigh(x)$, and otherwise,
\[
 \alpha(x, y) = 1 \wedge \frac{\pi(y) \, Q(y, x)}{\pi(x) \, Q(x, y)},
\]
noting that $R = \{(x, y): x \in \Xset, y \in \neigh(x)\}$. For all $x \in \Xset$, the MH algorithm thus proceeds with a proposal $y \sim Q(x, \cdot \,)$ which is accepted with probability
\[
 \alpha(x, y) = 1 \wedge \frac{\pi(y) \, Q(y, x)}{\pi(x) \, Q(x, y)},
\]
in which case the next state of the Markov chain is $y$. Otherwise, the proposal is rejected and the next state of the Markov chain is $x$.

In the auxiliary sampler, we define the directional neighbourhoods as follows: for all $x \in \Xset$, $\neigh_{-1}(x)$ and $\neigh_{+1}(x)$ are such that $\neigh_{-1}(x) \cup \neigh_{+1}(x) = \neigh(x)$ and either $\neigh_{-1}(x) \cap \neigh_{+1}(x) = \varnothing$ or $\neigh_{-1}(x) \cap \neigh_{+1}(x) = \{ x \}$. Typically, when the state-space is discrete, $x \notin \neigh(x)$ and $\neigh_{-1}(x) \cap \neigh_{+1}(x) = \varnothing$, as in our first example in \autoref{sec:example1}. When the state-space is continuous, $x \in \neigh(x)$ and $\neigh_{-1}(x) \cap \neigh_{+1}(x) = \{ x \}$, but $Q(x, \{x\}) = 0$, as in our second example in \autoref{sec:example2}. We also require that the directional neighbourhoods satisfy the following: for any $(x, y) \in \Xset^2$ and $\nu \in \{-1, +1\}$, $y \in \neigh_\nu(x)$ if and only if $x \in \neigh_{-\nu}(y)$, and for such a related couple $(x, y)$, $Q(x, \neigh_\nu(x)) > 0$ and $Q(y, \neigh_{-\nu}(y)) > 0$.

Consequently, \autoref{ass:sets} is satisfied and \autoref{prop:rev_derivative} holds, implying that $P_{\rev}$ is $\pi$-reversible, where
\[
 P_{\rev}(x, \d y) = Q_{\rev}(x, \d y) \, \alpha_{\rev}(x, y) + \delta_{x}(\d y) \int (1 - \alpha_{\rev}(x, u)) \, Q_{\rev}(x, \d u), \quad \text{for all $x \in \Xset$},
\]
$Q_{\rev}(x,\cdot\,) = (1/2) Q_{-1}(x, \cdot\,) + (1/2)Q_{+1}(x, \cdot\,)$ if $x \in \interior$ and $Q_{\rev}(x,\cdot\,) = (1/2)\delta_{x}+(1/2)Q(x,\cdot\,)$ otherwise, recalling the definitions of $Q_\nu$ and $\alpha_{\rev}$ in \autoref{sec:general_auxiliary}. For all $x \in \Xset$, the auxiliary sampler thus proceeds with a proposal $y \sim Q_{\rev}(x,\cdot\,)$ which is accepted with probability
\[
 \alpha_{\rev}(x, y) = 1 \wedge \frac{\pi(y) \, Q(y, x)}{\pi(x) \, Q(x, y)} \frac{Q(x, \neigh_{\nu}(x))}{Q(y, \neigh_{-\nu}(y))},
\]
for some $\nu \in \{-1, +1\}$, except if $(x, y) \notin R$, in which case $\alpha_{\rev}(x, y) = 0$, or if $y = x$, in which case $\alpha_{\rev}(x, y)) = 1$. As with the MH algorithm, if the proposal is accepted, the next state of the Markov chain is $y$; otherwise, the next state of the Markov chain is $x$.

\autoref{prop:generalized_lifted} also holds, meaning that $P_{\lifted}$ satisfies a skewed detailed balance (recall the definition of $P_{\lifted}$ in \autoref{sec:general_lifted}). When the current state of the Markov chain $(x, \nu) \in \Xset \times \{-1, +1\}$ is such that $x \in \interior$ or $x \in \bnd$ with $Q(x, \neigh_\nu(x)) = 1$, the lifted sampler proceeds with a proposal $y \sim Q_\nu(x, \cdot \,)$ which is accepted with probability
\[
 1 \wedge \frac{\pi(y) \, Q(y, x)}{\pi(x) \, Q(x, y)} \frac{Q(x, \neigh_{\nu}(x))}{Q(y, \neigh_{-\nu}(y))},
\]
 in which case the next state of the Markov chain is $(y, \nu)$. Otherwise, the proposal is rejected and the next state of the Markov chain is $(x, -\nu)$. If the current state of the Markov chain $(x, \nu) \in \Xset \times \{-1, +1\}$ is such that $x \in \bnd$ with $Q(x, \neigh_\nu(x)) = 0$, the next state of the Markov chain is automatically set to $(x, -\nu)$.

\subsection{Example 1: Partially-ordered finite state-space}\label{sec:example1}

In this example, we consider that $\Xset$ is a finite state-space which admits a partial order. This partial order will be exploited to induce directions to follow by the lifted sampler. We consider more precisely that $x = (x_1, \ldots, x_n)$, where each component $x_i$ can be of two types. Many statistical contexts fit within this framework, such as the modelling of binary data using networks or graphs and in variable selection. Indeed, for the former, $\Xset$ can be parameterized such that $\Xset=\{-1,+1\}^n$, where for example for an Ising model, $x_i \in \{-1, +1\}$ represents the state of a spin. For variable selection, $\Xset=\{0,1\}^n$ and $x_i\in\{0,1\}$ indicates whether or not the $i$-th covariate is included in the model employed.

The use of lifted algorithms to sample from distributions defined on partially-ordered finite state-spaces has been thoroughly studied in \cite{gagnon2023}. In particular, the Ising model has been studied. Here, we provide numerical results that are complementary to those presented there. The specific model that we study is a two-dimensional Ising model. For this model, the state-space $\Xset = (V_\eta, E_\eta)$ is a $\eta \times \eta$ square lattice regarded here as a square matrix in which each element takes either the value $-1$ or $+1$. We write each state as a vector as above: $x = (x_1, \ldots, x_n)$, where $n = \eta^2$. The states can be encoded as follows: the values of the components on the first line are $x_1,  \ldots, x_\eta$, those on the second line $x_{\eta + 1}, \ldots, x_{2\eta}$, and so on. The probability mass function is given by
\[
 \pi(x) = \frac{1}{Z} \exp\left(\sum_{i} \alpha_i x_i + \lambda \sum_{\langle i j \rangle} x_i x_j\right),
\]
where $\alpha_1, \ldots, \alpha_n \in \re$ and $\lambda > 0$ are fixed parameters, $Z$ is the normalizing constant and the notation $\langle i j \rangle$ indicates that sites $i$ and $j$ are nearest neighbours. The notion of neighbourhood on $(V_\eta, E_\eta)$ should not be confused with that on $\Xset$ on which the samplers rely. The neighbourhood of a site $i\in V_\eta$ comprises, when they exist, its North-South-East-West neighbours on the lattice. The parameter $\alpha := (\alpha_1, \ldots, \alpha_n)$ is often referred to as the \textit{external field} and $\lambda$ plays the role of a \textit{spatial correlation parameter}. The external field tends to polarize each spin, while the spatial correlation tends to create patches of identical spin states.

To simulate an Ising model, the MH algorithm typically proceeds by proposing to flip a single bit at each iteration (hence, $x \notin \neigh(x)$). More formally, the neighbourhood structure $\{\neigh(x)\}$ is given by $\neigh(x) = \{ y \in \Xset: \sum_i |x_i - y_i| = 2\}$. A uniform proposal distribution is often employed. This strategy has been shown to be often inefficient in \cite{zanella2019informed}, who proposed a generic approach to sample from distributions defined on discrete state-spaces. The approach leverages local target-distribution information and is thus less naive than a uniform approach. It possesses appealing properties and is referred to as \textit{locally balanced}. The terminology \textit{locally balanced} comes form the fact that, in the limit, when the state-space becomes larger and larger (but the neighbourhoods have a fixed size and proposed moves are thus local), there is no need for an accept-reject step anymore; the proposal distribution leaves the distribution $\pi$ invariant. The locally-balanced proposal distribution that we use in this example is the \textit{Barker proposal distribution}, with
\[
 q(x, y) = \frac{\pi(y)/\pi(x)}{1 + \pi(y)/\pi(x)}, \quad x \in \Xset, y \in \neigh(x),
\]
where the name is in reference to \cite{barker1965monte}'s acceptance probability choice. Note that the example fits within the framework of \autoref{sec:samplers_dominating} with the MH algorithm defined as in that section.

To explicitly define the directional neighbourhoods that are used in the lifted sampler, we leverage that $\Xset$ admits a partial order. Indeed, an inclusion-based partial order on $\Xset$ can be defined through a set:
\begin{align*}
\R := \left\{(x, y) \in \Xset \times \Xset : \{i\,:\,x_i = +1\} \subset
\{i\,:\,y_i = +1\}\right\}.
\end{align*}
Pairs $(x,y)\in\Xset^2$ with $x\neq y$ are said to be \textit{comparable} when either $(x,y)\in\R$ or $(y,x)\in\R$ and are said \textit{incomparable} otherwise. The existence of incomparable pairs $(x,y)$ with $x\neq y$ (those with the same number of $+1$ components) represents the difference with a totally-ordered set such as $\na$ or $\re$ in which every pair of different elements is comparable. We denote $x \prec y$ whenever $(x, y)\in\R$ and $x \neq y$. Hence, setting  $\neigh(x) = \{ y \in \Xset: \sum_i |x_i - y_i| = 2\}$ yields neighbourhoods with states $y$ that are all comparable with $x$. In the lifted sampler, we set $\neigh_{-1}(x) = \{y \in \neigh(x) : y \prec x\}$ and $\neigh_{+1}(x) = \{y \in \neigh(x) : x \prec y\}$. States $y$ in $\neigh_{+1}(x)$, for instance, are equal to $x$ except for one component which is flipped from $-1$ to $+1$. Note that the directional neighbourhoods satisfy \autoref{ass:sets}. Therefore, Propositions \ref{prop:rev_derivative} and \ref{prop:generalized_lifted} hold, and $P_{\rev}$ is $\pi$-reversible and $P_{\lifted}$ satisfies a skewed detailed balance, with the auxiliary and lifted samplers defined as in \autoref{sec:samplers_dominating}.

We are now ready to present simulation results. They are about a common problem in statistical physics which is to estimate the average \textit{magnetisation} of an Ising model, the magnetisation being defined as the mapping $x \mapsto \sum_{i=1}^n x_i$. The numerical results that we present are more specifically about the asymptotic variances of MCMC estimators of the expectation of the magnetisation, in fact, of the standardized version of the magnetisation, meaning that the expectation is subtracted from the mapping and the result of this subtraction is divided by the standard deviation. The standardized version is considered to allow comparable asymptotic variances when varying problem parameters, which is what we do in the numerical experiment, as described below. It also allows to fit within the framework of comparison between samplers described at the end of \autoref{sec:general_comparison}.

We consider a base target distribution for which $n = 10^2$ and the spatial correlation is moderate and more precisely $ \lambda = 0.5$. We set the $\alpha_i$ as follows:  $\alpha_i = -\mu + \epsilon_i$ if the column index is smaller than or equal to $\lfloor \eta / 2 \rfloor$ and $\alpha_i = \mu + \epsilon_i$ otherwise, where $\mu = 1$, the $\epsilon_i$ are independent uniform random variables on the interval $(-0.1, +0.1)$ and $\lfloor\, \cdot \,\rfloor$ is the floor function. We can thus think of the external field as a matrix with negative values on the left, and positive ones on the right. The base target is moderately rough, in the sense that it concentrates on a subset of the state-space, but with directional neighbourhoods that have a smoothly varying mass on this subset (implying similar $Q(x, \neigh_{-1}(x))$ and $Q(x, \neigh_{+1}(x))$).

The base target distribution described above is favourable for the lifted sampler with the locally-balanced proposal distribution which leverages variations in the target distribution and create persistent movement (as seen in \autoref{fig_results_Ising} by looking at the starting points on the left of the plots). In \autoref{fig_results_Ising} (a), we observe the impact of dealing with larger systems for a moderately rough target, by increasing $\eta$ from $10$ to $50$. Increasing $\eta$ in this case leads to longer paths along which the state-space can be explored, which is again favourable for lifted samplers. In \autoref{fig_results_Ising} (a), we observe that the lifted sampler scales better with the system size than the MH algorithm. When comparing the latter with the reversible counterpart to the lifted sampler, we observe a stable absolute difference of asymptotic variances, which translates into a diminishing relative difference. With \autoref{fig_results_Ising} (b), we consider an opposite situation which is increasingly unfavourable for the lifted sampler, with a fixed system size of $\eta = 10$ but an increasing $\mu$ yielding an increasingly rougher target distribution with an increasing level of concentration. As the level of concentration increases, the mass of the directional neighbourhoods becomes increasingly variable (implying increasingly different $Q(x, \neigh_{-1}(x))$ and $Q(x, \neigh_{+1}(x))$). Also, when the mass is concentrated on few configurations, it leaves not much room for persistent movement for the lifted sampler, and it thus loses its advantage. Such a situation corresponds to one where it is suspected that the lifted sampler will not offer a great performance. When the lifted sampler is at the mode and leaves it, it ``wastes'' an iteration because it tries continuing in the same direction whereas the MH sampler has the possibility to return to the mode the following iteration. In \autoref{fig_results_Ising} (b), we also observe an increasing difference in asymptotic variances between $P_{\rev}$ and $P_{\MH}$, with that of $P_{\rev}$ reaching the upper bound provided by \autoref{thm2} of $2\vara(f, P_{\MH}) + 1$. Note that the numerical results are based on 100 independent runs of 1,000,000 iterations for each algorithm and each value of $\eta$ and $\mu$, with burn-ins of 100,000.

\begin{figure}[ht]
\centering
$\begin{array}{cc}
    \hspace{-0mm} \includegraphics[width=0.50\textwidth]{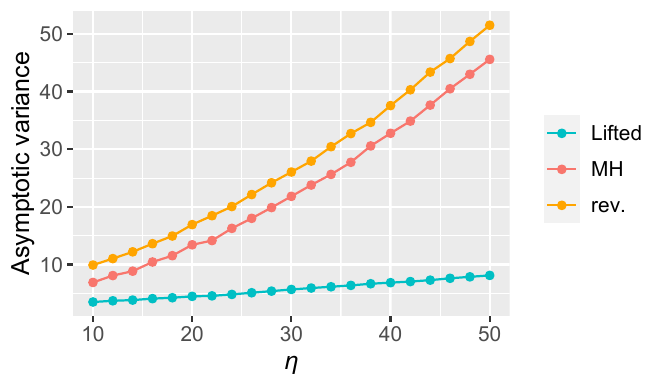} & \hspace{-4mm} \includegraphics[width=0.50\textwidth]{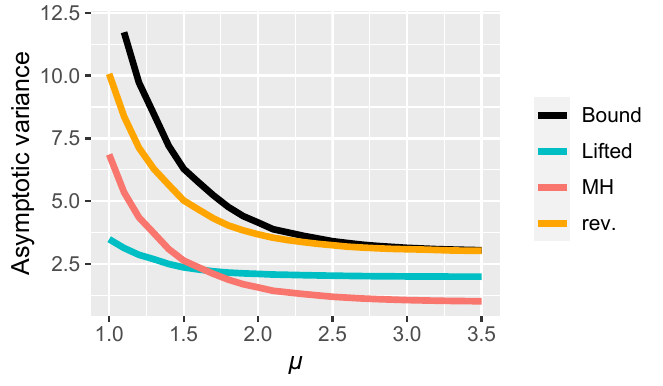} \cr
    \hspace{-11mm} \textbf{(a)} & \hspace{-12mm} \textbf{(b)}
\end{array}$
  \vspace{-2mm}
\caption{Asymptotic variance of the MCMC estimator of the standardized version of the mapping $x \mapsto \sum_{i=1}^n x_i$ for the MH algorithm, the lifted sampler and the reversible counterpart of the latter, all using the Barker proposal distribution when: (a) $\eta$ increases from 10 to 50 and the other parameters are kept fixed ($\mu = 1$ and $\lambda = 0.5$); (b) $\mu$ increases from 1 to 3.5 and the other parameters are kept fixed ($\eta = 10$ and $\lambda = 0.5$); in (b), the upper bound provided in \autoref{thm2} is also presented.}\label{fig_results_Ising}
\end{figure}
\normalsize

\subsection{Example 2: Totally-ordered continuous state-space}\label{sec:example2}

In this section, we study the simple case where $\Xset = \re$. When the MH sampler is the random walk Metropolis algorithm, the lifted version described in \autoref{sec:samplers_dominating} corresponds to the \textit{guided walk} \citep{gustafson1998guided} if $\neigh_{-1}(x) = (-\infty, x]$ and $\neigh_{+1}(x) = [x, +\infty)$ for all $x \in \Xset$. The guided walk has essentially the same computational cost as the random walk Metropolis algorithm, and has been proved to outperform the latter for any target distribution $\pi$ \citep[Theorem 7]{andrieu2021peskun}. This is a consequence of the symmetry of the normal distribution used in random walk Metropolis. This symmetry leads to a situation of perfect balance, as described in \autoref{sec:general_comparison}, where $Q(x, \neigh_{-1}(x)) = Q(x, \neigh_{+1}(x)) = 1 / 2$ for all $x \in \Xset$.

But what if the proposal distribution is not symmetrical? This is the case that we now study. In this case, one may suspect that the lifted sampler will not offer a great performance in situations of significant differences between $Q(x, \neigh_{-1}(x))$ and $Q(x, \neigh_{+1}(x))$. In our example, we consider that the proposal distribution is the continuous version of the Barker proposal distribution described in \autoref{sec:example1}. It has recently been proposed in \cite{livingstone2019robustness} and has the following density:
\[
 Q(x, y) = 2 \, \frac{\varphi_\sigma(y - x)}{1 + \exp\left(-(y - x) \, \nabla \log \pi(x)\right)}, \quad x \in \re, y \in \neigh(x) = \re,
\]
where $\varphi_\sigma$ is the density of a normal distribution with a mean of $0$ and a standard deviation of $\sigma > 0$, and $\nabla \log \pi$ is the gradient of the target log-density. The gradient injects skewness: the more substantial is $\nabla \log \pi(x)$, the more asymmetrical is the density $Q(x, \cdot \,)$. The scale parameter $\sigma$ is a tuning parameter. We study the natural case where the target distribution is a standard normal. Note that the example fits within the framework of \autoref{sec:samplers_dominating} with the MH algorithm defined as in that section.

 We present empirical results in \autoref{table:1} and \autoref{fig_results_Barker}. The directional neighbourhoods are as above: $\neigh_{-1}(x) = (-\infty, x]$ and $\neigh_{+1}(x) = [x, +\infty)$ for all $x \in \Xset$. The directional neighbourhoods thus satisfy \autoref{ass:sets}. Therefore, Propositions \ref{prop:rev_derivative} and \ref{prop:generalized_lifted} hold, and $P_{\rev}$ is $\pi$-reversible and $P_{\lifted}$ satisfies a skewed detailed balance, with the auxiliary and lifted samplers defined as in \autoref{sec:samplers_dominating}. In \autoref{table:1}, we present results for the MH sampler, its lifted version and the reversible counterpart of the latter. We more precisely present acceptance rates and asymptotic variances of the MCMC estimators of $\pi f$ when $f$ is the identity mapping $x \mapsto x$, for different values of $\sigma$. Note that the identity mapping is already standardized given that the target distribution is a standard normal. The values of $\sigma$ in \autoref{table:1} represent optimal values for the different samplers (at least according to our grid search). For the MH algorithm, the optimal value is $2.5$, associated with an acceptance rate of $62\%$. When using this algorithm, the asymptotic variance of the MCMC estimator of the mean is $1.94$. The lifted sampler suffers from instability among the mass of the directional neighbourhoods $Q(x, \neigh_{-1}(x))$ and $Q(x, \neigh_{+1}(x))$. To compensate for smaller acceptance rates, $\sigma$ has to be reduced. The optimal value for this sampler is $2.0$ and is associated with an acceptance rate of $46\%$. When using this algorithm, the asymptotic variance of the MCMC estimator of the mean is $2.31$. We observe in \autoref{fig_results_Barker} that the reason why the lifted sampler manages to offer a performance that is not so bad (despite significantly smaller $\sigma$ and acceptance rate) is persistent movement. Nevertheless, the impact of persistent movement is not significant enough. In addition, the computational cost associated to sampling from conditional distributions $Q_\nu(x, \cdot \,)$ and computing the normalizing constant of the latter is quite high in this case: runtime of the lifted sampler is about 130 times larger.\footnote{We do not claim optimality of our implementation, but we employed commonly used techniques for sampling from the conditional distributions $Q_\nu(x, \cdot \,)$ and computing the normalizing constant of the latter.} Regarding the reversible counterpart of the lifted sampler, the optimal value of $\sigma$ is $2.2$, which is in between those of the two other samplers. When using this value, the asymptotic variance is $4.08$, which is not so far below the bound provided in \autoref{thm2} of $4.99$. The smallest difference between the bound and the asymptotic variance among those calculated from \autoref{table:1} is $0.76$ when $\sigma = 2.5$. Note that the numerical results are based on one run of 1,000,000 iterations for each algorithm and each value of $\sigma$, started in stationarity.

\begin{table}[ht]
 \centering
\small
\begin{tabular}{l rr rr rr r}
\toprule
  &  \multicolumn{2}{c}{MH sampler} & \multicolumn{2}{c}{Lifted sampler} & \multicolumn{2}{c}{Reversible counterpart} & Upper bound  \cr
  \cmidrule(l){2-3} \cmidrule(l){4-5} \cmidrule(l){6-7}
  & Acc. rate & Asymp. var. & Acc. rate & Asymp. var. & Acc. rate & Asymp. var. & Asymp. var. \cr
\midrule
$\sigma = 2.0$ & $71\%$ & $2.10$ & $46\%$ & $2.31$ & $46\%$ & $4.17$ & $5.20$ \cr
$\sigma = 2.2$ & $67\%$ & $2.00$ & $43\%$ & $2.35$ & $43\%$ & $4.08$ & $4.99$ \cr
$\sigma = 2.5$ & $62\%$ & $1.94$ & $38\%$ & $2.47$ & $38\%$ & $4.13$ & $4.89$ \cr
\bottomrule
\end{tabular}
  \caption{Acceptance rates and asymptotic variances for a MH sampler with a Barker proposal distribution, its lifted version and the reversible counterpart of the latter, along with the upper bound on the asymptotic variances provided in \autoref{thm2}, for different values of $\sigma$.}\label{table:1}
\end{table}

\begin{figure}[ht]
\centering
$\begin{array}{cc}
\textbf{MH sampler with $\sigma = 2.5$} & \textbf{Lifted sampler with $\sigma = 2.0$} \cr
   \vspace{-0mm}\textbf{Asymptotic variance = 1.94} & \textbf{Asymptotic variance = 2.31} \cr
    \hspace{-2mm} \includegraphics[width=0.50\textwidth]{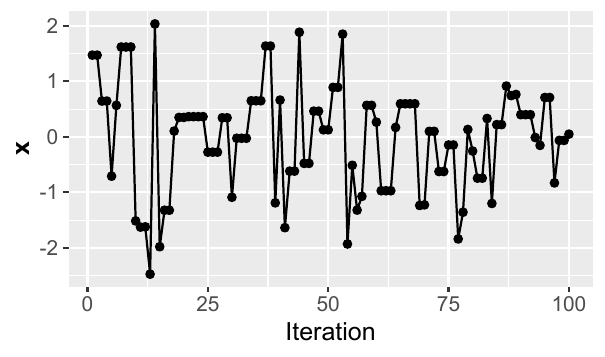} & \hspace{-4mm} \includegraphics[width=0.50\textwidth]{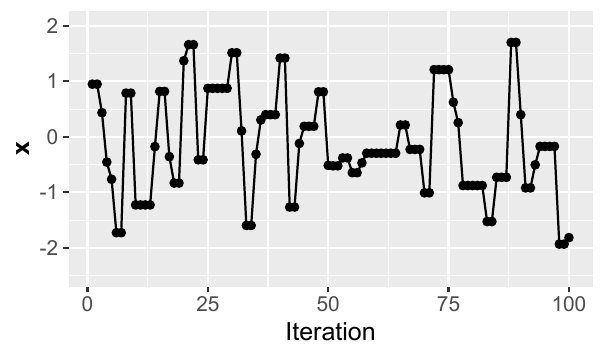}
\end{array}$
  \vspace{-2mm}
\caption{Trace plots for the MH algorithm on the left panel and the lifted sampler on the right panel, both with the Barker proposal distribution.}\label{fig_results_Barker}
\end{figure}
\normalsize

\section{Conclusion}\label{sec:conclusion}

In this paper, we leveraged the generalized MH algorithm of \cite{tierney1998note} to introduce a generalized lifted sampler. This allowed to understand how to construct a lifted sampler under a general framework in which the target and proposal distributions do not necessarily admit densities with respect to a common dominating measure. In the same way the generalized MH algorithm of \cite{tierney1998note} unified seemingly different reversible MCMC methods under the same framework, the generalized lifted sampler unifies seemingly different lifted samplers such as the non-reversible jump algorithm of \cite{gagnon2019NRJ} and the non-reversible simulated tempering of \cite{sakai2016irreversible}. Under this general framework, theoretical guarantees were provided: the generalized lifted sampler may produce MCMC estimators with larger asymptotic variances than the generalized MH algorithm, but the ratio of asymptotic variances is essentially bounded above by 2.

In \autoref{sec:optimality}, we highlight that different choices of directions may lead to lifted samplers that are significantly different, ranging from a lifted sampler which is worst than its MH counterpart to a lifted sampler which is better. This raises the question of the existence of an optimal choice of directions and, in the affirmative, of a theoretical guarantee that the optimal lifted sampler is better than its MH counterpart.

The focus of this article has been on the lifting technique in the classical sense, following \cite{andrieu2021peskun}, where the lifted sampler alternates between  two sub-probability measures $T_{-1}$ and $T_{+1}$ which satisfy the skewed detailed balance. The terminology \textit{lifting} is sometimes used in a more general context and refers to as a technique by which the state-space is \textit{lifted} to include auxiliary variables for performance improvement \citep{apers2021characterizing}. From this perspective, Hamiltonian Monte Carlo \citep{Duane1987, neal2011mcmc} and piecewise deterministic Monte Carlo methods \citep{vanetti2017piecewise, bouchard2018bouncy, 10.1214/18-AOS1715} are lifted samplers. For such classes of Monte Carlo methods, whether it is possible or not to establish theoretical guarantees as in our article is an open question. A challenge in these cases seems to be the identification of a natural and direct point of comparison, meaning MCMC algorithms that have a strong connection with the lifted samplers (as in our article).

Drawing inspiration from the recent work of \cite{ascolani2025fast}, we can extend the framework of \autoref{sec:general} to the case where the generalized auxiliary and lifted samplers use $D$ directions by exploiting a mixture structure, that is the Markov kernels correspond to a mixture of $D$ kernels, and \autoref{thm2} still holds. In \cite{ascolani2025fast}, a lifted sampler exploiting a mixture structure is introduced in the context of Bayesian finite mixture models and an upper bound is obtained on the asymptotic variances, in the same spirit as \autoref{thm2}. In the special case where the mixture model in \cite{ascolani2025fast} has two components, $D = 1$ and their sampler is a special case of the generalized lifted sampler introduced in \autoref{sec:general_lifted}. In this case, our generic upper bound in \autoref{thm2} allows to retrieve that in \cite{ascolani2025fast} derived within their specific framework.

\section*{Acknowledgments}

The authors thank two anonymous referees for constructive comments that led to an improved article. Philippe Gagnon acknowledges support from NSERC (Natural Sciences and Engineering Research Council of Canada) and CANSSI (Canadian Statistical Sciences Institute). Florian Maire acknowledges support from NSERC.

\bibliographystyle{rss}
\bibliography{references}

\appendix

\section{Proofs}\label{sec:proofs}

In this section, we present proofs of results in the same order as these results appeared in the paper. We start with the proof of \autoref{lemma:varphi}, followed by the proofs of Propositions \ref{prop:rev_derivative} and \ref{prop:generalized_lifted}. Before presenting the proof of \autoref{prop:rev_derivative}, we present and prove one lemma (\autoref{lemma3}) that will be useful. Before presenting the proof of \autoref{prop:generalized_lifted}, we provide the proof that a skewed detailed balance implies that the product measure $\pi \otimes \mathcal{U}\{-1, +1\}$ is an invariant distribution. We finish the section with \autoref{Lemma:ordering} and its proof.

\begin{proof}[Proof of \autoref{lemma:varphi}]
We prove the results in the same order as they appear in \autoref{lemma:varphi}.

Proof of Result 1. Let $r>0$ and $\eps>0$, then $\varphi(r+\eps)=r\varphi(1/(r+\eps))+\eps\varphi(1/(r+\eps))$. Assume that $\varphi$ is decreasing, that is $r<r'$ implies that $\varphi(r) > \varphi(r')$. Then, $\varphi(1/(r+\eps)) > \varphi(1/r)$ and
$$
\varphi(r+\eps)=r\varphi(1/(r+\eps))+\eps\varphi(1/(r+\eps))>r\varphi(1/r)=\varphi(r),
$$
which leads to a contradiction.

Proof of Result 2. Let $r>0$. Take $\eps>0$ and note that $\varphi(r+\eps)-\varphi(r)\geq 0$ (by Result 1). Moreover,
$$
\varphi(r+\eps)=(r+\eps)\varphi(1/(r+\eps))\leq (r+\eps)\varphi(1/r)\leq \varphi(r)+\eps,
$$
using again Result 1 and the fact that $\varphi \leq 1$. Thus, $\varphi(r)\leq \varphi(r+\eps)\leq \varphi(r)+\eps$ and $\varphi$ is continuous from the right. A symmetric argument shows that $\varphi$ is also continuous from the left, thus $\varphi$ is continuous.

Proof of Result 3. We have that $0 \leq \varphi(r) = r \varphi (1 / r) \leq r$, which allows to conclude by taking $r \downarrow 0$.

Proof of Result 4. The null function verifies $\varphi(r) = r \varphi(1 / r) = 0$ for all $r \in (0, \infty)$. Now consider that $\varphi$ is not the null function. Then, $\varphi(r) > 0$ for some $r$ and $\varphi(r') > 0$ for all $r' \geq r$ by Result 1. For any $r'$ as small as we want, we also have $\varphi(r') > 0$ given that $\varphi(r') = r' \varphi(1 / r') > 0$. Result 1, thus allows to confirm that $\varphi(r) > 0$ for all $r \in (0, \infty)$.

 Proof of Result 5. If $\varphi$ is the null function, then $0 = \varphi(ab)\geq \varphi(a)\varphi(b) = 0$ for any $a>0, b>0$. Now consider that $\varphi$ is not the null function; thus, it is a positive function. We consider two cases. First, we consider that $a\leq 1, b > 0$. Thus, $1/ab\geq 1/b$ and $\varphi(1/ab)\geq \varphi(1/b)$ from Result 1. Then, using that $\varphi(1/a)\leq 1$, we have that
$$
\varphi(ab)=ab\varphi(1/ab)\geq ab\varphi(1/b)\geq ab\varphi(1/b)\varphi(1/a)=\varphi(a)\varphi(b).
$$
Now, consider that $a> 1, b > 0$. Given that $b=\varphi(b)/\varphi(1/b)$ and $a=\varphi(a)/\varphi(1/a)$, we have that
$$
\varphi(ab)=\frac{\varphi(a)}{\varphi(1/a)}\frac{\varphi(b)}{\varphi(1/b)}\varphi(1/ab)=\varphi(a)\varphi(b)\frac{\varphi(1/ab)}{\varphi(1/a)\varphi(1/b)}\geq \varphi(a)\varphi(b)\,,
$$
where the last inequality follows from the fact that $\varphi(1/ab)\geq \varphi(1/a)\varphi(1/b)$ using the result of the first case: $\varphi(a'b') \geq \varphi(a')\varphi(b')$ with $a' = 1/a \leq 1$ and $b' = 1 / b > 0$.

Proof of Result 6. We have that $1/ab\geq 1/a $ so $\varphi(1/ab)\geq \varphi(1/a)$ (Result 1), but $\varphi(ab)=ab\varphi(1/ab)\geq ab\varphi(1/a)=b\varphi(a)$.
\end{proof}

Before presenting the proof of \autoref{prop:rev_derivative}, we present and prove one lemma that will be useful. Recall the definitions presented in Sections \ref{sec:general_MH} and \ref{sec:general_auxiliary}.

\begin{Lemma}\label{lemma3}
 Suppose that \autoref{ass:sets} holds. Let $g: \Xset^2 \to \re$ be such that
 \[
  g(x, y) = \ind((x, y) \in R_{-1}  \cap R_{+1}^\mathsf{c}) + \ind((x, y) \in R_{+1} \cap R_{-1}^\mathsf{c}),
 \]
  for all $(x, y) \in \Xset^2$. Then, $g = 1$ $\mu$-almost everywhere and $\mu^T$-almost everywhere. In particular, for any measurable function $\vartheta:\Xset^2\to\re$,
 \[
  \int \vartheta \, \d\mu = \sum_{\nu \in \{-1, +1\}} \int_{R_\nu \cap R_{-\nu}^\mathsf{c}} \vartheta \, \d\mu,
 \]
 and
  \[
  \int \vartheta \, \d\mu^T = \sum_{\nu \in \{-1, +1\}} \int_{R_\nu \cap R_{-\nu}^\mathsf{c}} \vartheta \, \d\mu^T.
  \]
\end{Lemma}

\begin{proof}
 For $(x, y) \in \Xset^2$,
   \begin{equation*}
g(x,y) =\left\{
\begin{array}{lll}
  0 & \text{if} & (x,y)\in (R_{-1}^\mathsf{c} \cap R_{+1}^\mathsf{c}) \cup (R_{+1} \cap R_{+1}^\mathsf{c}) \cup (R_{-1} \cap R_{-1}^\mathsf{c}) \cup (R_{+1} \cap R_{-1}),\\
   2 & \text{if} & (x,y)\in R_{-1} \cap R_{-1}^\mathsf{c} \cap R_{+1} \cap R_{+1}^\mathsf{c}, \\
   1 & \text{otherwise}.
\end{array}
\right.
\end{equation*}

Clearly,
\[
 \mu(R_{-1} \cap R_{-1}^\mathsf{c} \cap R_{+1} \cap R_{+1}^\mathsf{c}) = 0.
\]
Also, by the union bound,
\begin{align*}
&\mu( (R_{-1}^\mathsf{c} \cap R_{+1}^\mathsf{c}) \cup (R_{+1} \cap R_{+1}^\mathsf{c}) \cup (R_{-1} \cap R_{-1}^\mathsf{c}) \cup (R_{+1} \cap R_{-1})) \cr
&\quad \leq \mu(R_{-1}^\mathsf{c} \cap R_{+1}^\mathsf{c}) + \mu(R_{+1} \cap R_{+1}^\mathsf{c}) + \mu(R_{-1} \cap R_{-1}^\mathsf{c}) + \mu(R_{+1} \cap R_{-1}).
\end{align*}
Clearly,
\[
 \mu(R_{+1} \cap R_{+1}^\mathsf{c}) = \mu(R_{-1} \cap R_{-1}^\mathsf{c}) = 0.
\]
By definition,
\[
 \mu(R_{-1} \cap R_{+1}) = \int \ind((x, y) \in R_{-1} \cap R_{+1}) \, \mu(\d x, \d y) = \int \pi(\d x) \, Q(x, \neigh_{-1}(x) \cap \neigh_{+1}(x)),
\]
and
\[
 \mu(R_{-1}^\mathsf{c} \cap R_{+1}^\mathsf{c}) = \int \ind((x, y) \in R_{-1}^\mathsf{c} \cap R_{+1}^\mathsf{c}) \, \mu(\d x, \d y) = \int \pi(\d x) \, Q(x, \neigh_{-1}(x)^\mathsf{c} \cap \neigh_{+1}(x)^\mathsf{c}),
\]
given that $R_{\nu}^\mathsf{c} = \{(x, y) \in \Xset^2: y \in \neigh_\nu(x)\}^\mathsf{c} = \{(x, y) \in \Xset^2: y \notin \neigh_\nu(x)\}$. But, we have that $Q(x, \neigh_{-1}(x) \cap \neigh_{+1}(x)) = 0$ and $Q(x, \neigh_{-1}(x)^\mathsf{c} \cap \neigh_{+1}(x)^\mathsf{c}) = 0$ for all $x \in \Xset$ by \autoref{lemma:P_rev}. This implies that $g = 1$ $\mu$-almost everywhere. We also have that $g = 1$ $\mu^T$-almost everywhere given that $g$ is symmetric (using that $y \in \neigh_\nu(x)$ \textit{if and only if} $x \in \neigh_{-\nu}(y)$).
 \end{proof}

\begin{proof}[Proof of \autoref{prop:rev_derivative}]
We start by proving that, for all $(x, y) \in R$, $0 < r_{\rev}(x, y) < \infty$ and $r_{\rev}(x, y) = 1 / r_{\rev}(y, x)$, which guarantees that everything in the following is well defined. We have that, for all $(x, y) \in R$, $r_{\rev}(x, y) = r(x, y) \, \beta(x, y)$ and $0 < r(x, y) < \infty$ and $0 < \beta(x, y) < \infty$, the former following from \autoref{Prop:Tierney} and the latter from the definition of $\beta$ (see \autoref{sec:general_auxiliary}). Also, for all $(x, y) \in R$, $r_{\rev}(x, y) = r(x, y) \, \beta(x, y) = (1 / r(y, x)) (1 / \beta(y, x)) = 1 / r_{\rev}(y, x)$ by \autoref{Prop:Tierney} and the fact that $\beta(x, y) = Q(x, \neigh_{\nu}(x)) / Q(y, \neigh_{-\nu}(y)) = (Q(y, \neigh_{-\nu}(y)) / Q(x, \neigh_{\nu}(x)))^{-1}$ when $(x, y) \in R_{\nu} \cap R_{-\nu}^\mathsf{c} \cap \Delta^\mathsf{c} \cap R$ for some $\nu \in \{-1, +1\}$, in which case $(y, x) \in R_{-\nu} \cap R_{\nu}^\mathsf{c} \cap \Delta^\mathsf{c} \cap R$, and $\beta(x, y) = 1$ otherwise, meaning that $(x, y) \notin R_{\nu} \cap R_{-\nu}^\mathsf{c} \cap \Delta^\mathsf{c} \cap R$ for all $\nu \in \{-1, +1\}$ which is equivalent to $(y, x) \notin R_{\nu} \cap R_{-\nu}^\mathsf{c} \cap \Delta^\mathsf{c} \cap R$ for all $\nu \in \{-1, +1\}$.

To prove that for all $(x, y) \in R$,
\[
 r_{\rev}(x, y) = \frac{\d\mu_{\rev}^T}{\d\mu_{\rev}}(x, y) = r(x, y) \, \beta(x, y),
\]
we will prove that
\[
 \int_{R} \phi \, r_{\rev} \, \d \mu_{\rev} = \int_{R} \phi \, \d \mu_{\rev}^T,
\]
for any bounded measurable function $\phi:\Xset^2\to\re$. This will allow to conclude that the restriction of the measures $\mu_{\rev}$ and $\mu_{\rev}^T$ to $R$ are mutually absolutely continuous given that $r_{\rev}$ is positive on $R$ (recall the definition of $r_{\rev}$ in \autoref{sec:general_auxiliary}).

We use the following decomposition:
\[
\Xset\times\Xset=\{\interior\times\interior\}\cup\{\interior\times\bnd\}\cup\{\bnd\times\interior\}\cup \{\bnd\times\bnd\}.
\]
We use the linearity of the integral and analyse each part separately.

For any bounded measurable function $\phi:\Xset^2\to\re$, by the definition of $\mu_{\rev}$,
\begin{align*}
\int_{\{\interior \times \interior\}\cap R} \phi \, r_{\rev} \, \d \mu_{\rev}
&=\sum_{\nu \in \{-1, +1\}}\int_{\{\interior \times \interior\}\cap R} \phi(x, y) \, r_{\rev}(x, y) \, \pi(\d x) \, (1/2) \, Q_{\nu}(x,\d y) \cr
&=\sum_{\nu \in \{-1, +1\}}\int_{\{\interior \times \interior\}\cap R \cap R_\nu \cap R_{-\nu}^\mathsf{c} \cap \Delta^\mathsf{c}}\phi(x, y) \, r_{\rev}(x, y) \, \pi(\d x) \, (1/2) \, Q_{\nu}(x, \d y),
\end{align*}
using that $\pi(\d x) \, Q_{\nu}(x, \d y)$ assigns null measure to $\{\interior \times \interior\} \cap (R_\nu \cap R_{-\nu}^\mathsf{c})^\mathsf{c}$ and $\{\interior \times \interior\} \cap \Delta$, for $\nu \in\{-1, +1\}$. Indeed,
\begin{align*}
& \int_{\{\interior \times \interior\}\cap (R_\nu^\mathsf{c} \cup R_{-\nu})} \pi(\d x) \, Q_{\nu}(x,\d y) \cr
 &\quad= \int_{\{\interior \times \interior\}} \ind((x, y) \in R_\nu^\mathsf{c} \cup R_{-\nu}) \, \pi(\d x) \, Q_{\nu}(x,\d y) \cr
 &\quad\leq \int_{\{\interior \times \interior\}} \ind((x, y) \in R_\nu^\mathsf{c}) \, \pi(\d x) \, Q_{\nu}(x,\d y)  + \int_{\{\interior \times \interior\}} \ind((x, y) \in R_{-\nu}) \, \pi(\d x) \, Q_{\nu}(x,\d y) \cr
 &\quad= \int_{\interior} \pi(\d x) \, Q_{\nu}(x, \neigh_\nu(x)^\mathsf{c} \cap \interior)  + \int_{\interior} \pi(\d x) \, Q_{\nu}(x, \neigh_{-\nu}(x) \cap \interior) \cr
 &\quad\leq \int_{\interior} \pi(\d x) \, Q_{\nu}(x, \neigh_\nu(x)^\mathsf{c})  + \int_{\interior} \pi(\d x) \, Q_{\nu}(x, \neigh_{-\nu}(x) ) \cr
 &\quad= \int_{\interior} \pi(\d x) \, \frac{Q(x, \neigh_\nu(x)^\mathsf{c} \cap \neigh_\nu(x))}{Q(x, \neigh_\nu(x)}  + \int_{\interior} \pi(\d x) \, \frac{Q(x, \neigh_{-\nu}(x) \cap \neigh_\nu(x))}{Q(x, \neigh_\nu(x)} = 0,
\end{align*}
and
\begin{align*}
 \int_{\{\interior \times \interior\} \cap \Delta} \pi(\d x) \, Q_{\nu}(x,\d y) = \int_{\{\interior \times \interior\}} \ind(x = y) \, \pi(\d x) \, Q_{\nu}(x,\d y) &=\int_{\interior} \pi(\d x) \, Q_{\nu}(x,\{x\}) \cr
 &=\int_{\interior} \pi(\d x) \, \frac{Q(x,\{x\})}{Q(x, \neigh_\nu(x))} = 0,
\end{align*}
using the union bound, the definition of $R_\nu$, that $R_{\nu}^\mathsf{c} = \{(x, y) \in \Xset^2: y \in \neigh_\nu(x)\}^\mathsf{c} = \{(x, y) \in \Xset^2: y \notin \neigh_\nu(x)\}$, and that $Q(x, \neigh_{-1}(x) \cap \neigh_{+1}(x)) = 0$ and $Q(x, \{x\}) = 0$ for all $x \in \Xset$ by \autoref{lemma:P_rev}. Therefore, for $\nu \in\{-1, +1\}$ and for any measurable function $\vartheta:\Xset^2 \rightarrow \re$,
\begin{align}\label{eqn:null_measure}
  \int_{\{\interior \times \interior\} \cap (R_\nu \cap R_{-\nu}^\mathsf{c})^\mathsf{c}}  \vartheta(x, y) \, \pi(\d x) \, Q_\nu(x, \d y) = \int_{\{\interior \times \interior\} \cap \Delta}  \vartheta(x, y) \, \pi(\d x) \, Q_\nu(x, \d y) = 0.
\end{align}

Also,
\begin{align*}
  &\int_{\{\interior \times \interior\}\cap R \cap R_\nu \cap R_{-\nu}^\mathsf{c} \cap \Delta^\mathsf{c}} \phi(x, y) \, r_{\rev}(x, y) \, \pi(\d x) \, (1/2) \, Q_{\nu}(x, \d y) \cr
   &\quad = \int_{\{\interior \times \interior\}\cap R \cap R_\nu \cap R_{-\nu}^\mathsf{c} \cap \Delta^\mathsf{c}} \frac{\phi(x, y)}{2 Q(y, \neigh_{-\nu}(y))} \, r(x, y) \, \mu(\d x, \d y) \cr
  &\quad= \int_{\{\interior \times \interior\}\cap R \cap R_\nu \cap R_{-\nu}^\mathsf{c} \cap \Delta^\mathsf{c}} \frac{\phi(x, y)}{2 Q(y, \neigh_{-\nu}(y))} \, \mu^T(\d x, \d y) \cr
  &\quad =  \int_{\{\interior \times \interior\}\cap R \cap R_\nu \cap R_{-\nu}^\mathsf{c} \cap \Delta^\mathsf{c}} \phi(x, y) \,  \pi(\d y) \, (1/2) \, Q_{-\nu}(y, \d x) \cr
  &\quad=  \int_{\{\interior \times \interior\}\cap R} \phi(x, y) \,  \pi(\d y) \, (1/2) \, Q_{-\nu}(y, \d x).
\end{align*}
In the first equality, we used the definition of $r_{\rev}$. In the second equality, we used that $\int_R \phi \, r \, \d\mu = \int_R \phi \, \d\mu^T$ for all $\phi$ bounded measurable (\autoref{Prop:Tierney}). Note that $\phi(x, y) / (2 Q(y, \neigh_{-\nu}(y))$ is not necessarily bounded, but we can limit ourselves to the case where $Q(y, \neigh_{-\nu}(y))$ is bounded from below by an arbitrarily small constant and use a limiting argument with the monotone convergence theorem (where we separate the cases i) $\phi(x, y) \geq 0$ and ii) $\phi(x, y) < 0$). In the third equality, we used that $(x, y) \in R_\nu$ if and only if $(y, x) \in R_{-\nu}$. In the fourth equality, we used the same arguments as before to obtain an equality by omitting the intersection with $\Delta^\mathsf{c}$ and $R_\nu \cap R_{-\nu}^\mathsf{c}$, recalling that $(x, y) \in R_\nu$ if and only if $(y, x) \in R_{-\nu}$. Therefore,
\begin{align*}
 \int_{\{\interior \times \interior\}\cap R} \phi \, r_{\rev} \, \d \mu_{\rev}
&=\sum_{\nu \in \{-1, +1\}}\int_{\{\interior \times \interior\}\cap R} \phi(x, y) \, r_{\rev}(x, y) \, \pi(\d x) \, (1/2) \, Q_{\nu}(x,\d y) \cr
&= \sum_{\nu \in \{-1, +1\}}\int_{\{\interior \times \interior\}\cap R} \phi(x, y) \,  \pi(\d y) \, (1/2) \, Q_{-\nu}(y, \d x) \cr
&= \int_{\{\interior \times \interior\}\cap R} \phi \, \d \mu^T_{\rev}.
\end{align*}

Using similar arguments as above, we obtain that
\begin{align*}
\int_{\{\interior \times \bnd\}\cap R} \phi \, r_{\rev} \, \d \mu_{\rev}
&=\sum_{\nu \in \{-1, +1\}}\int_{\{\interior \times \bnd\}\cap R \cap R_\nu \cap R_{-\nu}^\mathsf{c} \cap \Delta^\mathsf{c}}\frac{\phi(x, y)}{2 Q(y, \neigh_{-\nu}(y))} \, r(x, y) \, \mu(\d x, \d y).
\end{align*}
But, under \autoref{ass:sets}, for any $(x, y) \in \{\interior \times \bnd\}\cap R_\nu$, we have that $Q(x, \neigh_\nu(x)) > 0$, $Q(y, \neigh_{-\nu}(y)) > 0$ and $Q(y, \neigh_{-1}(y)) Q(y, \neigh_{+1}(y)) = 0$, which, together with $1 = Q(y, \Xset) = Q(y, \neigh_{-1}(y)) + Q(y, \neigh_{+1}(y))$, implies that $Q(y,\neigh_{-\nu}(y)) = 1$. Therefore,
\begin{align*}
&\int_{\{\interior \times \bnd\}\cap R} \phi \, r_{\rev} \, \d \mu_{\rev} \cr
&\quad=\sum_{\nu \in \{-1, +1\}}\int_{\{\interior \times \bnd\}\cap R \cap R_\nu \cap R_{-\nu}^\mathsf{c} \cap \Delta^\mathsf{c}}\frac{\phi(x, y)}{2} \, \mu^T(\d x, \d y) \cr
&\quad= \sum_{\nu \in \{-1, +1\}}\int_{\{\interior \times \bnd\}\cap R \cap R_\nu \cap R_{-\nu}^\mathsf{c} \cap \Delta^\mathsf{c}} \phi(x, y) \,  \pi(\d y) \, (1/2) \, Q(y, \d x) \cr
&\quad= \int_{\{\interior \times \bnd\}\cap R} \phi(x, y) \left(\ind((x, y) \in R_{-1} \cap R_{+1}^\mathsf{c}) + \ind((x, y) \in R_{+1} \cap R_{-1}^\mathsf{c})\right) \pi(\d y) \, (1/2) \, Q(y, \d x) \cr
&\quad= \int_{\{\interior \times \bnd\}\cap R} \phi(x, y) \, \pi(\d y) \, (1/2) \, Q(y, \d x) \cr
&\quad= \int_{\{\interior \times \bnd\}\cap R} \phi(x, y) \,  \pi(\d y) \left((1/2) \, Q(y, \d x) + (1/2) \, \delta_{x}(\d y)\right) \cr
&\quad= \int_{\{\interior \times \bnd\}\cap R} \phi \, \d \mu^T_{\rev}.
\end{align*}
In the first equality, we used that $\int_R \phi \, r \, \d\mu = \int_R \phi \, \d\mu^T$ for all $\phi$ bounded measurable (\autoref{Prop:Tierney}). In the third equality, we used that $\pi(\d y) \, Q(y, \d x)$ assigns null measure to $\Delta$ and an argument similar as in \eqref{eqn:null_measure} (recall that $Q(y, \{y\}) = 0$ for all $y \in \Xset$ by \autoref{lemma:P_rev}). In the fourth equality, we applied \autoref{lemma3}.

Symmetrically,
\begin{align*}
\int_{\{\bnd \times \interior\}\cap R} \phi \, r_{\rev} \, \d \mu_{\rev}
&=\int_{\{\bnd \times \interior\}\cap R }\frac{\phi(x, y)}{2} r_{\rev}(x, y) \, \mu(\d x, \d y) \cr
&=\int_{\{\bnd \times \interior\}\cap R \cap \Delta^\mathsf{c}}\frac{\phi(x, y)}{2} r_{\rev}(x, y) \, \mu(\d x, \d y) \cr
&=\sum_{\nu \in \{-1,+1\}} \int_{\{\bnd \times \interior\}\cap R \cap R_\nu \cap R_{-\nu}^\mathsf{c} \cap \Delta^\mathsf{c}}\frac{\phi(x, y)}{2}\frac{Q(x, \neigh_{\nu}(x))}{Q(y, \neigh_{-\nu}(y))} \, r(x, y) \, \mu(\d x, \d y) \cr
&=\sum_{\nu \in \{-1,+1\}} \int_{\{\bnd \times \interior\}\cap R \cap R_\nu \cap R_{-\nu}^\mathsf{c} \cap \Delta^\mathsf{c} }\frac{\phi(x, y)}{2} \frac{1}{Q(y, \neigh_{-\nu}(y))} \, \mu^T(\d x, \d y) \cr
&=\sum_{\nu \in \{-1,+1\}}\int_{\{\bnd \times \interior\}\cap R \cap R_\nu \cap R_{-\nu}^\mathsf{c} \cap  \Delta^\mathsf{c}} \phi(x, y) \,\pi(\d y) \, (1/2) Q_{-\nu}(y, \d x) \cr
&=\sum_{\nu \in \{-1,+1\}}\int_{\{\bnd \times \interior\}\cap R }\phi(x, y) \,\pi(\d y) \, (1/2) Q_{-\nu}(y, \d x) \cr
&=\int_{\{\bnd \times \interior\}\cap R }\phi(x, y) \, \pi(\d y)  \left((1/2) Q_{-1}(y, \d x) + (1/2) Q_{+1}(y, \d x)\right) \cr
&= \int_{\{\bnd \times \interior\}\cap R} \phi \, \d \mu^T_{\rev}.
\end{align*}
In the second equality, we used that $\mu(\d x, \d y) = \pi(\d x) \, Q(x, \d y)$ assigns null measure to $\Delta$ and an argument similar as in \eqref{eqn:null_measure}. In the third equality, we applied \autoref{lemma3} and used the definition of the function $r_{\rev}$. In the fourth equality, we used that, similarly as in the previous case, $Q(x,\neigh_{\nu}(x)) = 1$ for any $(x, y) \in \{\bnd \times \interior\}\cap R_\nu$. We also used that $\int_R \phi \, r \, \d\mu = \int_R \phi \, \d\mu^T$ for all $\phi$ bounded measurable (that is \autoref{Prop:Tierney}; again we can use a limiting argument with the monotone convergence theorem). In the fifth equality, we used that $(x, y) \in R_\nu$ if and only if $(y, x) \in R_{-\nu}$. In the sixth equality, we used that $\pi(\d y) \, Q_{-\nu}(y, \d x)$ assigns null measure to $\{\bnd \times \interior\} \cap (R_\nu \cap R_{-\nu}^\mathsf{c})^\mathsf{c}$ and $\{\bnd \times \interior\} \cap \Delta$ and a similar argument as in \eqref{eqn:null_measure} (recall that $(x, y) \in R_\nu$ if and only if $(y, x) \in R_{-\nu}$).

Finally, using the definition of $r_{\rev}$ and $\mu_{\rev}$,
\begin{align*}
\int_{\{\bnd \times \bnd\}\cap R} \phi \, r_{\rev} \, \d \mu_{\rev}
&=\int_{\{\bnd \times \bnd\}\cap R \cap \Delta }\phi(x, y) \, r(x, y) \, \pi(\d x) \, (1/2)\, \delta_{x}(\d y) \cr
&\quad + \int_{\{\bnd \times \bnd\}\cap R \cap \Delta^\mathsf{c} }\phi(x, y) \, r_{\rev}(x, y) \, \pi(\d x) \, (1/2)\, Q(x, \d y),
\end{align*}
because
\[
 \int_{\{\bnd \times \bnd\}\cap R \cap \Delta^\mathsf{c} }\phi(x, y) \, r_{\rev}(x, y) \, \pi(\d x) \, (1/2)\, \delta_{x}(\d y) = 0,
\]
and
\[
  \int_{\{\bnd \times \bnd\}\cap R \cap \Delta }\phi(x, y) \, r_{\rev}(x, y) \, \pi(\d x) \, (1/2)\, Q(x, \d y) = 0,
\]
given that $\mu(\d x, \d y) = \pi(\d x) \, Q(x, \d y)$ assigns null measure to $\Delta$ by \autoref{lemma:P_rev} and using an argument similar as in \eqref{eqn:null_measure}. Using that $r(x, y) = 1$ when $y = x$,
\begin{align*}
 \int_{\{\bnd \times \bnd\}\cap R \cap \Delta }\phi(x, y) \, r(x, y) \, \pi(\d x) \, (1/2)\, \delta_{x}(\d y) &= \int_{\{\bnd \times \bnd\}\cap R \cap \Delta }\phi(x, y) \, \pi(\d x) \, (1/2)\, \delta_{x}(\d y) \cr
 &= \int_{\{\bnd \times \bnd\}\cap R \cap \Delta }\phi(x, y) \, \pi(\d y) \, (1/2)\, \delta_{y}(\d x).
\end{align*}

Also,
\begin{align*}
& \int_{\{\bnd \times \bnd\}\cap R \cap \Delta^\mathsf{c} }\phi(x, y) \, r_{\rev}(x, y) \, \pi(\d x) \, (1/2) \, Q(x, \d y) \cr
&\quad =  \int_{\{\bnd \times \bnd\}\cap R \cap \Delta^\mathsf{c} }\phi(x, y) \, r_{\rev}(x, y)  \, (1/2) \, \mu(\d x, \d y) \cr
&\quad = \sum_{\nu \in \{-1, +1\}}  \int_{\{\bnd \times \bnd\}\cap R \cap R_\nu \cap R_{-\nu}^\mathsf{c} \cap \Delta^\mathsf{c} }\phi(x, y) \, r_{\rev}(x, y) \,  (1/2) \, \mu(\d x, \d y) \cr
&\quad = \sum_{\nu \in \{-1, +1\}}  \int_{\{\bnd \times \bnd\}\cap R \cap R_\nu \cap R_{-\nu}^\mathsf{c} \cap \Delta^\mathsf{c} }\phi(x, y) \, r(x, y) \,  (1/2) \, \mu(\d x, \d y) \cr
&\quad = \sum_{\nu \in \{-1, +1\}}  \int_{\{\bnd \times \bnd\}\cap R \cap R_\nu \cap R_{-\nu}^\mathsf{c} \cap \Delta^\mathsf{c} }\phi(x, y) \, (1/2)\, \mu^T(\d x, \d y) \cr
&\quad =  \int_{\{\bnd \times \bnd\}\cap R \cap \Delta^\mathsf{c} }\phi(x, y) \, (1/2)\, \mu^T(\d x, \d y) \cr
 &\quad = \int_{\{\bnd \times \bnd\}\cap R \cap \Delta^\mathsf{c} }\phi(x, y) \, \pi(\d y) \, (1/2)\, Q(y, \d x).
\end{align*}
 In the second and fifth equalities, we applied \autoref{lemma3}. In the third equality, we used that for all $(x, y) \in \{\bnd \times \bnd\}\cap R \cap R_\nu \cap R_{-\nu}^\mathsf{c} \cap \Delta^\mathsf{c}$,  $Q(x, \neigh_{\nu}(x)) = 1$ and $Q(y, \neigh_{-\nu}(y)) = 1$ (similarly as in the previous cases). Therefore, for all $(x, y) \in \{\bnd \times \bnd\}\cap R \cap R_\nu \cap R_{-\nu}^\mathsf{c} \cap \Delta^\mathsf{c}$, $\beta(x, y) = 1$. In the fourth equality, we used that $\int_R \phi \, r \, \d\mu = \int_R \phi \, \d\mu^T$ for all $\phi$ bounded measurable (\autoref{Prop:Tierney}).

 Putting those results together yields
\begin{align*}
 \int_{\{\bnd \times \bnd\}\cap R} \phi \, r_{\rev} \, \d \mu_{\rev} &= \int_{\{\bnd \times \bnd\}\cap R \cap \Delta }\phi(x, y) \, \pi(\d y) \, (1/2)\, \delta_{y}(\d x) \cr
 & \quad + \int_{\{\bnd \times \bnd\}\cap R \cap \Delta^\mathsf{c} }\phi(x, y) \, \pi(\d y) \, (1/2)\, Q(y, \d x)
 = \int_{\{\bnd \times \bnd\}\cap R} \phi \, \d \mu_{\rev}^T,
\end{align*}
using similar arguments as above. This concludes the proof that $r_{\rev}(x, y) = (\d\mu_{\rev}^T / \d\mu_{\rev})(x, y) = r(x, y) \, \beta(x, y)$ for all $(x, y) \in R$.
\end{proof}

Before presenting the proof of \autoref{prop:generalized_lifted}, we provide the proof that a skewed detailed balance implies that the product measure $\pi \otimes \mathcal{U}\{-1, +1\}$ is an invariant distribution (recall the definition of skewed detailed balance presented in \autoref{sec:general_lifted}). Considering an event $A \in \Xalg$ and $\nu' \in \{-1, +1\}$,
\begin{align*}
 &\sum_{\nu \in \{-1, +1\}} \int \pi(\d x) \, \frac{1}{2} \, P_{\lifted}((x, \nu), A \times \nu') \cr
  &\quad = \frac{1}{2} \sum_{\nu \in \{-1, +1\}}   \left(\ind(\nu' = \nu) \int \pi(\d x) \, T_{\nu}(x, A) + \ind(\nu' = -\nu) \int \pi(\d x) \int_A \delta_{x}(\d y) \, (1 - T_{\nu}(x, \Xset)) \right) \cr
  &\quad = \frac{1}{2} \sum_{\nu \in \{-1, +1\}}   \left(\ind(\nu' = \nu) \int_A \pi(\d x) \, T_{-\nu}(x, \Xset) + \ind(\nu' = -\nu) \int_A \pi(\d x) \, (1 - T_{\nu}(x, \Xset)) \right) \cr
  &\quad = \frac{1}{2} \int_A \pi(\d x) \left( T_{-\nu'}(x, \Xset) + 1 -  T_{-\nu'}(x, \Xset)\right) = \frac{1}{2} \int_A \pi(\d x),
\end{align*}
where the skewed detailed balance is used in the second equality to obtain
\[
 \int \pi(\d x) \, T_{\nu}(x, A) = \int_A \pi(\d x) \, T_{-\nu}(x, \Xset).
\]
Therefore, the probability to reach $A \times \nu'$ in one step, if the chain is in stationarity (under $\pi \otimes \mathcal{U}\{-1, +1\}$), is equal to the probability of $A \times \nu'$ under the stationary distribution $\pi \otimes \mathcal{U}\{-1, +1\}$.

\begin{proof}[Proof of \autoref{prop:generalized_lifted}]
 As mentioned in \autoref{sec:general_lifted}, to prove that $P_{\lifted}$ satisfies a skewed detailed balance, we need to show that the two following measures on the product space $(\Xset^2,\Xalg\otimes \Xalg)$ are equal:
$$
\mu_{+1}(A\times B) = \int_A \pi(\d x) \, T_{+1}(x, B) \quad \text{and} \quad \mu_{-1}(A\times B) = \int_B \pi(\d x) \, T_{-1}(x, A),
$$
which ensures that $P_{\lifted}$ leaves the distribution $\pi \otimes \mathcal{U}\{-1, +1\}$ invariant. To achieve this, we prove that for any bounded measurable function $\phi: \Xset^2 \rightarrow \re$,
\[
 \int \phi \, \d\mu_{+1} = \int \phi \, \d\mu_{-1}.
\]

Let us define
\[
 G_{+1} := \{(x, y) \in \Xset^2: x \in \interior \quad \text{or} \quad x \in \bnd \text{ with }  Q(x, \neigh_{+1}(x)) = 1\},
\]
and
\[
 G_{-1} := \{(x, y) \in \Xset^2: y \in \interior \quad \text{or} \quad y \in \bnd \text{ with }  Q(y, \neigh_{-1}(y)) = 1\}.
\]

  We have that
\begin{align*}
 \int \phi \, \d\mu_{+1} &= \int_{G_{+1}} \phi(x, y) \, \alpha_{\rev}(x, y) \, \pi(\d x) \, Q_{+1}(x, \d y) \cr
 &= \int_{G_{+1} \cap R \cap R_{+1} \cap R_{-1}^\mathsf{c} \cap \Delta^\mathsf{c}} \phi(x, y) \, \alpha_{\rev}(x, y) \, \pi(\d x) \, Q_{+1}(x, \d y).
\end{align*}
In the first equality, we used the definitions of $\mu_{+1}$ and $T_{+1}$. In the second equality, we used that $\alpha_{\rev}(x, y) = 0$ when $(x, y) \notin R$, that $\pi(\d x) \, Q_{+1}(x, \d y)$ assigns null measure to $G_{+1} \cap (R_{+1} \cap R_{-1}^\mathsf{c})^\mathsf{c}$ and $G_{+1} \cap \Delta$, as we now prove, and a similar argument as in \eqref{eqn:null_measure}. Let us recall that it has been shown in \autoref{sec:general_lifted} that, if $(x, \nu)$ is such that $x \in \interior$ or $x \in \bnd$ with $Q(x, \neigh_\nu(x)) = 1$, then $Q_{\nu}(x, \neigh_\nu(x)^\mathsf{c}) = 0$, $Q_{\nu}(x, \neigh_{-\nu}(x)) = 0$ and $Q_{\nu}(x, \{x\}) = 0$. Let
\[
 G_{+1}^{x} := \{x \in \Xset: x \in \interior \quad \text{or} \quad x \in \bnd \text{ with }  Q(x, \neigh_{+1}(x)) = 1\}.
\]
Recall the definitions of $R_\nu$ and $\Delta$ in \autoref{sec:general_auxiliary}. We have that
\begin{align}\label{eqn:arg_prop_generalized_lifted}
 &\int_{G_{+1} \cap (R_{+1} \cap R_{-1}^\mathsf{c})^\mathsf{c}} \pi(\d x) \, Q_{+1}(x, \d y) \cr &\quad =\int_{G_{+1}} \ind((x, y) \in R_{+1}^\mathsf{c} \cup R_{-1}) \, \pi(\d x) \, Q_{+1}(x, \d y) \cr
 &\quad \leq \int_{G_{+1}} \ind((x, y) \in R_{+1}^\mathsf{c}) \, \pi(\d x) \, Q_{+1}(x, \d y) + \int_{G_{+1}} \ind((x, y) \in R_{-1}) \, \pi(\d x) \, Q_{+1}(x, \d y) \cr
 &\quad =\int_{G_{+1}^{x}}  \pi(\d x) \, Q_{+1}(x, \neigh_{+1}(x)^\mathsf{c}) + \int_{G_{+1}^{x}} \pi(\d x) \, Q_{+1}(x, \neigh_{-1}(x)) = 0,
\end{align}
given that $R_{\nu}^\mathsf{c} = \{(x, y) \in \Xset^2: y \in \neigh_\nu(x)\}^\mathsf{c} = \{(x, y) \in \Xset^2: y \notin \neigh_\nu(x)\}$. Also,
\[
 \int_{G_{+1} \cap \Delta} \pi(\d x) \, Q_{+1}(x, \d y) =  \int_{G_{+1}} \ind(x = y) \, \pi(\d x) \, Q_{+1}(x, \d y) = \int_{G_{+1}^{x}} \pi(\d x) \, Q_{+1}(x, \{x\}) = 0.
\]

Let $G_{+1}^1 :=  \{(x, y) \in \Xset^2: x \in \interior\}$ and $G_{+1}^2 :=  \{(x, y) \in \Xset^2: x \in \bnd \text{ with }  Q(x, \neigh_{+1}(x)) = 1\}$ with $G_{+1}^1 \cup  G_{+1}^2 = G_{+1}$ and $G_{+1}^1 \cap  G_{+1}^2 = \varnothing$.
We have that
\begin{align*}
 &\int_{G_{+1} \cap R \cap R_{+1} \cap R_{-1}^\mathsf{c} \cap \Delta^\mathsf{c}} \phi(x, y) \, \alpha_{\rev}(x, y) \, \pi(\d x) \, Q_{+1}(x, \d y) \cr
 & \quad = \int_{G_{+1}^1 \cap R \cap R_{+1} \cap R_{-1}^\mathsf{c} \cap \Delta^\mathsf{c}} \phi(x, y) \,  \alpha_{\rev}(x, y) \, \pi(\d x) \, Q_{+1}(x, \d y) \cr
 &\qquad + \int_{G_{+1}^2 \cap R \cap R_{+1} \cap R_{-1}^\mathsf{c} \cap \Delta^\mathsf{c}} \phi(x, y) \,  \alpha_{\rev}(x, y) \, \pi(\d x) \, Q_{+1}(x, \d y) \cr
  & \quad = \int_{G_{+1}^1 \cap R \cap R_{+1} \cap R_{-1}^\mathsf{c} \cap \Delta^\mathsf{c}} \phi(x, y) \,  \alpha_{\rev}(x, y) \, \pi(\d x) \, (Q_{+1}(x, \d y) + Q_{-1}(x, \d y))  \cr
 &\qquad + \int_{G_{+1}^2 \cap R \cap R_{+1} \cap R_{-1}^\mathsf{c} \cap \Delta^\mathsf{c}} \phi(x, y) \,  \alpha_{\rev}(x, y) \, \pi(\d x) \, (Q_{+1}(x, \d y) + \delta_x(\d y)) \cr
   & \quad = \int_{G_{+1}^1 \cap R \cap R_{+1} \cap R_{-1}^\mathsf{c} \cap \Delta^\mathsf{c}} \phi(x, y) \,  \alpha_{\rev}(x, y) \, \pi(\d x) \, 2 Q_{\rev}(x, \d y)  \cr
 &\qquad + \int_{G_{+1}^2 \cap R \cap R_{+1} \cap R_{-1}^\mathsf{c} \cap \Delta^\mathsf{c}} \phi(x, y) \,  \alpha_{\rev}(x, y) \, \pi(\d x) \, 2 Q_{\rev}(x, \d y) \cr
 & \quad = \int_{G_{+1} \cap R \cap R_{+1} \cap R_{-1}^\mathsf{c} \cap \Delta^\mathsf{c}} \phi(x, y) \,  \alpha_{\rev}(x, y)  \, 2 \mu_{\rev}(\d x, \d y),
\end{align*}
using the definitions of $Q_{\rev}$ and $\mu_{\rev}$. In the second equality, we used that
\[
  \int_{G_{+1}^1 \cap R \cap R_{+1} \cap R_{-1}^\mathsf{c} \cap \Delta^\mathsf{c}} \phi(x, y) \,  \alpha_{\rev}(x, y) \, \pi(\d x) \, Q_{-1}(x, \d y) = 0,
\]
and
\[
 \int_{G_{+1}^2 \cap R \cap R_{+1} \cap R_{-1}^\mathsf{c} \cap \Delta^\mathsf{c}} \phi(x, y) \,  \alpha_{\rev}(x, y) \, \pi(\d x) \, \delta_x(\d y) = 0,
\]
where the former follows from the fact that the measure $\pi(\d x) \, Q_{-1}(x, \d y)$ assigns null measure to the set $R_{+1} \cap R_{-1}^\mathsf{c}$ (recall \autoref{lemma:P_rev} and \eqref{eqn:arg_prop_generalized_lifted}) and latter follows directly from the definition of $\Delta^\mathsf{c}$.

Therefore,
\begin{align*}
 \int \phi \, \d\mu_{+1} &= \int_{G_{+1} \cap R \cap R_{+1} \cap R_{-1}^\mathsf{c} \cap \Delta^\mathsf{c}} \phi(x, y) \,  \alpha_{\rev}(x, y)  \, 2 \mu_{\rev}(\d x, \d y) \cr
  &= \int_{G_{+1} \cap R \cap R_{+1} \cap R_{-1}^\mathsf{c} \cap \Delta^\mathsf{c}} \phi(x, y) \, \varphi(1 / r_{\rev}(x, y)) \, 2 \, r_{\rev}(x, y) \, \mu_{\rev}(\d x, \d y) \cr
  &= \int_{G_{+1} \cap R \cap R_{+1} \cap R_{-1}^\mathsf{c} \cap \Delta^\mathsf{c}} \phi(x, y) \, \varphi(1 / r_{\rev}(x, y)) \, 2 \, \mu_{\rev}^T(\d x, \d y) \cr
  &= \int_{G_{-1} \cap R \cap R_{+1} \cap R_{-1}^\mathsf{c} \cap \Delta^\mathsf{c}} \phi(x, y) \, \varphi(r_{\rev}(y, x)) \, 2 \, \pi(\d y) \, Q_{\rev}(y, \d x) \cr
   &= \int_{G_{-1} \cap R \cap R_{+1} \cap R_{-1}^\mathsf{c} \cap \Delta^\mathsf{c}} \phi(x, y) \, \alpha_{\rev}(y, x) \, \pi(\d y) \, Q_{-1}(y, \d x) \cr
   &= \int_{G_{-1}} \phi(x, y) \, \alpha_{\rev}(y, x) \, \pi(\d y) \, Q_{-1}(y, \d x) = \int \phi \, \d\mu_{-1}.
\end{align*}
In the second equality, we used the definitions of $\alpha_{\rev}$ and $\varphi$. In the third and fourth equalities, we used \autoref{prop:rev_derivative}. In the fourth equality, we also used that, for all $(x, y) \in G_{+1} \cap R \cap R_{+1} \cap R_{-1}^\mathsf{c} \cap \Delta^\mathsf{c}$, $Q(x, \neigh_{+1}(x)) > 0$ and $Q(y, \neigh_{-1}(y)) > 0$ under \autoref{ass:sets}, implying that $G_{+1} \cap R \cap R_{+1} \cap R_{-1}^\mathsf{c} \cap \Delta^\mathsf{c} = G_{-1} \cap R \cap R_{+1} \cap R_{-1}^\mathsf{c} \cap \Delta^\mathsf{c}$. Finally, in the fourth equality, we used the definition of $\mu_{\rev}^T$. In the fifth and sixth equalities, we used symmetrical arguments as above by recalling that $(x, y) \in R_{\nu}$ if and only if $(y, x) \in R_{-\nu}$. In the final equality, we used the definitions of $\mu_{-1}$ and $T_{-1}$. This concludes the proof.
\end{proof}

We finish this section with a lemma which is used in the proof of \autoref{thm2}. Before presenting it, we introduce notation. For two functions $f, g: \Xset \rightarrow \re$, let $\pscal{f}{g} := \int f \, g \, \d\pi$, $\|f\| := \left[\int f^2 \, \d\pi\right]^{1/2}$ and $L^2(\pi) := \{f: \|f\| < \infty\}$. For $f \in L^2(\pi)$, we can consider without loss of generality that it has a mean of 0. Let $L_0^2(\pi) := \{f: \|f\| < \infty \text{ and } \pi f = 0\}$. \autoref{Lemma:ordering} below can be seen as an extension of Lemma 33 in \cite{andrieu2018uniform}: in the latter, $(I - P)$, with $I$ being the identity operator and $P$ a $\pi$-reversible kernel, is assumed to be invertible, whereas we do not assume this, yet the same conclusion is obtained. The proof is strongly inspired from that of Lemma 33 in \cite{andrieu2018uniform}, but it uses the $\lambda$-asymptotic variance, instead of using directly the asymptotic variance.

\begin{Lemma}\label{Lemma:ordering}
  Let $P_1$ and $P_2$ be two $\pi$-reversible Markov kernels such that, for any $x \in \Xset$ and $A \in \Xalg$, $P_2(x, A \setminus \{x\}) \geq \alpha P_1(x, A \setminus \{x\})$ for some $\alpha \in (0, 1]$. Then, for any $f\in L_0^2(\pi)$,
  $$
  \vara_\lambda(f, P_2)\leq \frac{1}{\alpha}\vara_\lambda(f, P_1)+\left(\frac{1}{\alpha}-1\right)\|f\|^2, \quad \lambda \in [0, 1).
  $$
Also, by taking the limit $\lambda\uparrow 1$, the same inequality holds for the asymptotic variances given that $  \lim_{\lambda\to 1}\vara_\lambda(f, P) = \vara(f, P)$, whether the limit is finite or not, for any $\pi$-reversible kernel $P$.
\end{Lemma}
% rendu ici
\begin{proof}
As mentioned in \cite{andrieu2021peskun}, the $\lambda$-asymptotic variance can be written as
\begin{equation}
\label{eq4}
\vara_\lambda(f, P)=2\pscal{f}{(I-\lambda P)^{-1} f} - \| f\|^2.
\end{equation}
The fact that $P$ is Markov guarantees that $(I-\lambda P)$ has a bounded inverse on $L^2(\pi)$. Indeed, $\pscal{f}{(I-\lambda P)^{-1} f}=(1/\lambda)\pscal{f}{\{(1/\lambda)I- P\}^{-1} f}$ and $(1/\lambda)I- P$ is invertible given that $1/\lambda>1$ and $1$ is an upper bound of the spectrum of $P$. Thus, using Lemma 16 in \cite{andrieu2016establishing}, we have that
$$
\vara_\lambda(f, P)=2\sup\{2\pscal{f}{g}-\pscal{g}{(I-\lambda P)g}\,:\,g\in L_0^2(\pi)\}-\| f\|^2.
$$
For convenience, let us define $\D_{\lambda P}(f):=\pscal{f}{(I-\lambda P)f}$, the generalization of the Dirichlet form $\D_{P}(f)=\pscal{f}{(I-P)f}$, so that
$$
\vara_\lambda(f,P)=2\sup\{2\pscal{f}{g}-\D_{\lambda P}(g)\,:\,g\in L_0^2(\pi)\}-\|f\|^2.
$$

We now establish that $P_2(x, A \setminus \{x\}) \geq \alpha P_1(x, A \setminus \{x\})$ for all $x \in \Xset$ and $A \in \Xalg$ implies that $\D_{\lambda P_2}(g)\geq \alpha \D_{\lambda P_1}(g)$ for any $g\in L_0^2(\pi)$:
\begin{align*}
  \D_{\lambda P_2}(g) &=\lambda \pscal{g}{(I-P_2) g}+(1-\lambda)\|g\|^2 \cr
  &\geq \alpha \lambda \pscal{g}{(I-P_1) g}+(1-\lambda)\|g\|^2 \cr
  &=\alpha\lambda \|g\|^2-\alpha\lambda \pscal{g}{P_1g}+(1-\lambda)\|g\|^2 \cr
  &=\alpha \|g\|^2-\alpha\pscal{g}{\lambda P_1g}+(1-\lambda)\|g\|^2+\alpha(\lambda-1)\|g\|^2 \cr
  &=\alpha\pscal{g}{(I-\lambda P_1)g}+(1-\lambda)(1-\alpha)\|g\|^2 \cr
  &=\alpha\D_{\lambda P_1}(g)+(1-\lambda)(1-\alpha)\|g\|^2\geq \alpha\D_{\lambda P_1}(g).
\end{align*}
We can now unfold the same argument as in the proof of Lemma 33 in \cite{andrieu2018uniform} to establish that
\begin{align*}
\sup\{2\pscal{f}{g}-\D_{\lambda P_2}(g)\,:\,g\in L_0^2(\pi)\}&\leq \sup\{2\pscal{f}{g}-\alpha \D_{\lambda P_1}(g)\,:\,g\in L_0^2(\pi)\} \cr
&=\frac{1}{\alpha} \sup\{2\pscal{f}{\alpha g}- \D_{\lambda P_1}(\alpha g)\,:\,g\in L_0^2(\pi)\} \cr
&=\frac{1}{\alpha} \sup\{2\pscal{f}{g}- \D_{\lambda P_1}(g)\,:\,g\in L_0^2(\pi)\}\,.
\end{align*}
The inequality is obtained using \eqref{eq4}.
\end{proof}

\section{No Peskun ordering in all generality}\label{sec:noPeskun}

In this section, we show using an example that it is not possible to obtain a Peskun ordering in all generality when considering the general definition of lifted samplers in \cite{andrieu2021peskun}. The Peskun ordering that we consider is the generalized version given in \autoref{Lemma:ordering}. Let us consider two $\pi$-reversible transition kernels $P_1$ and $P_2$ operating on the same space $(\Xset, \Xalg)$. We say that there exists a (generalized) Peskun ordering between $P_1$ and $P_2$ if there exists some $\alpha \in (0, 1]$ such that, for any $x \in \Xset$ and $A \in \Xalg$,
\[
 P_2(x, A \setminus \{x\}) \geq \alpha P_1(x, A \setminus \{x\}).
\]
Conversely, we say that $P_1$ and $P_2$ do not admit a Peskun ordering if for all $\alpha \in (0, 1]$ (as small as we want), there exists $A \in \Xalg$ (that may depend on $\alpha$) such that $P_2(x, A \setminus \{x\}) < \alpha P_1(x, A \setminus \{x\})$ for some $x \in \Xset$.

We now present the example. Let $\Xset = \re$ and $\pi = \mathcal{N}(0, 1)$. Let $Q_{+1}(x, \cdot \,) = \mathcal{N}(x, \sigma^2)$ and $Q_{-1}(x, \cdot \,) = \mathcal{N}(x, 1 / \sigma^2)$ with $\sigma \in (0, 1)$. Let $Q = (1 / 2) Q_{-1} + (1 / 2) Q_{+1}$. Note that the example fits within the framework of \autoref{sec:samplers_dominating} with $P_{\MH}$ defined as in that section. For any $A \in \Xalg$ such that $0 \notin A$, we have that
\begin{align*}
 P_{\MH}(0, A) &= \int_A \left(\frac{1}{2} \, Q_{-1}(0, y) + \frac{1}{2} \, Q_{+1}(0, y)\right) \left(1 \wedge \exp\left(-\frac{y^2}{2}\right)\right) \d y \cr
 &= \int_A \left(\frac{1}{2} \, Q_{-1}(0, y) + \frac{1}{2} \, Q_{+1}(0, y)\right) \exp\left(-\frac{y^2}{2}\right) \d y,
\end{align*}
using the symmetry of $Q$.

We can define a lifted version similarly as in \autoref{sec:general_lifted}. For all $x \in \Xset$ and $\nu \in \{-1, +1\}$, let
\[
T_{\nu}(x, \d y) := Q_{\nu}(x, y) \left(1 \wedge \frac{\pi(y) \, Q_{-\nu}(y,x)}{\pi(x) \, Q_\nu(x,y)}\right) \d y.
\]
Note that the definition does not involve a split of $Q$ using directional neighbourhoods, and thus \autoref{ass:sets} cannot be verified. As in the proof of \autoref{prop:generalized_lifted}, we can show that $\mu_{+1}(A\times B) = \int_A \pi(\d x) \, T_{+1}(x, B)$ and $\mu_{-1}(A\times B) = \int_B \pi(\d x) \, T_{-1}(x, A)$ are equal, which implies that
\[
P_{\lifted}((x, \nu), \d(y, \nu')) := T_{\nu}(x, \d y) \, \delta_\nu(\d\nu') + \delta_{x}(\d y) \, \delta_{-\nu}(\d\nu') \, (1 - T_{\nu}(x, \Xset))
\]
leaves the distribution $\pi \otimes \mathcal{U}\{-1, +1\}$ invariant.

The reversible counterpart to the lifted sampler is defined as
\[
 P_{\rev}(x, \d y) := \frac{1}{2} \, T_{+1}(x, \d y) + \frac{1}{2} \, T_{-1}(x, \d y) +  \delta_{x}(\d y) \left(1 - \frac{1}{2} \, T_{+1}(x, \Xset) - \frac{1}{2} \, T_{-1}(x, \Xset)\right).
\]
This definition can be seen as being similar to that in \autoref{sec:general_auxiliary} and, like in that section, it can be proved that $P_{\rev}$ is $\pi$-reversible.

For the same $A \in \Xalg$ as above with $0 \notin A$, we have that
\begin{align*}
 P_{\rev}(0, A) &=  \int_A \frac{1}{2} \, Q_{+1}(0, y) \left(1 \wedge \exp\left(-\frac{y^2}{2}\right) \frac{Q_{-1}(y, 0)}{Q_{+1}(0, y)}\right) \d y \cr
 &\qquad + \int_A \frac{1}{2} \, Q_{-1}(0, y) \left(1 \wedge \exp\left(-\frac{y^2}{2}\right) \frac{Q_{+1}(y, 0)}{Q_{-1}(0, y)}\right) \d y.
\end{align*}
We have that
\[
 \exp\left(-\frac{y^2}{2}\right) \frac{Q_{-1}(y, 0)}{Q_{+1}(0, y)} = \sigma^2 \exp\left(\frac{y^2}{2} \left(\frac{1}{\sigma^2} - \sigma^2 - 1\right)\right),
\]
and
\[
 \exp\left(-\frac{y^2}{2}\right) \frac{Q_{+1}(y, 0)}{Q_{-1}(0, y)} = \frac{1}{\sigma^2} \exp\left(-\frac{y^2}{2} \left(\frac{1}{\sigma^2} + 1 - \sigma^2\right)\right).
\]
Note that
\[
 \frac{1}{\sigma^2} + 1 - \sigma^2 > 0,
\]
because $\sigma \in (0, 1)$. We can choose $\sigma$ small enough so that
\[
 \frac{1}{\sigma^2} - \sigma^2 - 1 > 0.
\]
We choose $\sigma$ accordingly. Therefore, we can choose $A = \{y \in \re: |y| > k\}$ with $k$ a large enough constant, so that, for all $y \in A$, we have that
\[
 \exp\left(-\frac{y^2}{2}\right) \frac{Q_{-1}(y, 0)}{Q_{+1}(0, y)} \geq 1 \quad \text{and} \quad \exp\left(-\frac{y^2}{2}\right) \frac{Q_{+1}(y, 0)}{Q_{-1}(0, y)} \leq 1.
\]
Note that $0 \notin A$. With such a set $A$, we have that
\[
 P_{\rev}(0, A) =  \int_A \frac{1}{2} \, Q_{+1}(0, y) \, \d y + \int_A \frac{1}{2} \, Q_{+1}(0, y) \exp\left(-\frac{y^2}{2}\right) \d y.
\]

Let $\alpha \in (0, 1]$. We have that
\begin{align*}
 &P_{\rev}(0, A) - \alpha P_{\MH}(0, A) \cr
 &\quad= \int_A \frac{1}{2} \, Q_{+1}(0, y) \, \d y + \int_A \frac{1}{2} \, Q_{+1}(0, y) \exp\left(-\frac{y^2}{2}\right) \d y \cr
 &\qquad - \alpha \int_A \left(\frac{1}{2} \, Q_{-1}(0, y) + \frac{1}{2} \, Q_{+1}(0, y)\right) \exp\left(-\frac{y^2}{2}\right) \d y \cr
 &\quad= \frac{1}{2} \int_A Q_{+1}(0, y) \left(1 + \exp\left(-\frac{y^2}{2}\right) - \alpha \exp\left(-\frac{y^2}{2}\right) \frac{Q_{-1}(0, y)}{Q_{+1}(0, y)} - \alpha \exp\left(-\frac{y^2}{2}\right)\right) \d y \cr
 &\quad\leq \frac{1}{2} \int_A Q_{+1}(0, y) \left(1 + \exp\left(-\frac{y^2}{2}\right) - \alpha \exp\left(-\frac{y^2}{2}\right) \frac{Q_{-1}(0, y)}{Q_{+1}(0, y)}\right) \d y.
\end{align*}

Note that
\[
 \exp\left(-\frac{y^2}{2}\right) \frac{Q_{-1}(0, y)}{Q_{+1}(0, y)} = \exp\left(-\frac{y^2}{2}\right) \frac{Q_{-1}(y, 0)}{Q_{+1}(0, y)} = \sigma^2 \exp\left(\frac{y^2}{2} \left(\frac{1}{\sigma^2} - \sigma^2 - 1\right)\right).
\]
Therefore, regardless the value of $\alpha$, we could have chosen $k$ such that
\[
 1 + \exp\left(-\frac{y^2}{2}\right) < \alpha \exp\left(-\frac{y^2}{2}\right) \frac{Q_{-1}(0, y)}{Q_{+1}(0, y)}.
\]
Consequently, $P_{\rev}(0, A) - \alpha P_{\MH}(0, A) < 0$, which concludes the example.

\section{Optimality of our bound on the asymptotic variances}\label{sec:optimality}

In this section, we show using an example that the bound $\vara(f, P_{\lifted}) \leq \vara(f, P_{\rev}) \leq 2\vara(f, P_{\MH}) + \var[f(X)]$ in \autoref{thm2} is essentially optimal, in the sense that it is essentially not possible to obtain a better bound without additional assumptions.

Let $\mathcal{Z}=\{(0,0),(0,1),(1,1)\}\subset\mathbb{Z}^2$ which we identify with $\Xset=\{1,2,3\}$. Let us define the neighbourhood structure:
$$
\neigh(1)=\{2\}\,,\quad \neigh(2)=\{1,3\}\,,\quad \neigh(3)=\{2\}\,.
$$
It can be seen as following from the North-South-East-West neighborhood structure on $\mathbb{Z}^2$ by considering the intersection with $\mathcal{Z}$. Let $\pi$ be a probability on $\mathcal{X}$ with $\pi\{1\}=\pi\{3\}=(1-\eps)/2$ and $\pi\{2\}=\eps$, for some $\eps\in(0, 1/2)$. Let $P_{\MH}$ be the MH kernel defined as in \autoref{sec:samplers_dominating} with $Q(x,\cdot\,)$ given by the uniform distribution on $\neigh(x)$. Let us define $\xi := \eps/(1-\eps)$. We have that
$$
P_{\MH}=\begin{pmatrix}
          1-\xi & \xi & 0 \\
          1/2 & 0 & 1/2 \\
          0 & \xi & 1-\xi
        \end{pmatrix}\,.
$$

Let us define the directional neighbourhoods:
\begin{align*}
 &\neigh_{-1}(1) = \varnothing\,,\quad \neigh_{+1}(1) =  \{2\}, \cr
 &\neigh_{-1}(2) = \{1,3\}\,,\quad \neigh_{+1}(2) = \varnothing, \cr
 &\neigh_{-1}(3) = \varnothing\,,\quad \neigh_{+1}(3) = \{2\}.
\end{align*}
They can be seen as following from a split of the North-South-East-West neighbourhoods into North-West and South-East neighbourhoods. It can be readily verified that the directional neighbourhoods satisfy \autoref{ass:sets}. Let $P_{\lifted}$ be the lifted kernel defined as in \autoref{sec:samplers_dominating}. We have that
$$
P_{\lifted}=\begin{pmatrix}
              0 & \xi & 0 & 1-\xi & 0& 0\\
              0 & 0 & 0 & 0 & 1 & 0 \\
              0 & \xi & 0 & 0 & 0 & 1-\xi \\
              1 & 0 & 0 & 0 & 0 & 0\\
              0 & 0 & 0 & 1/2 & 0 & 1/2\\
              0 & 0 & 1 & 0 & 0 & 0
            \end{pmatrix} \,,
$$
where the lines and columns correspond to $(1,+1), (2,+1),(3,+1),(1,-1),(2,-1),(3,-1)$ in this order.

An illustration of the Markov chains is presented in \autoref{fig_example}.
\begin{figure}[ht]
\centering
\includegraphics[width=0.50\textwidth, trim={5cm 3cm 5cm 3cm}]{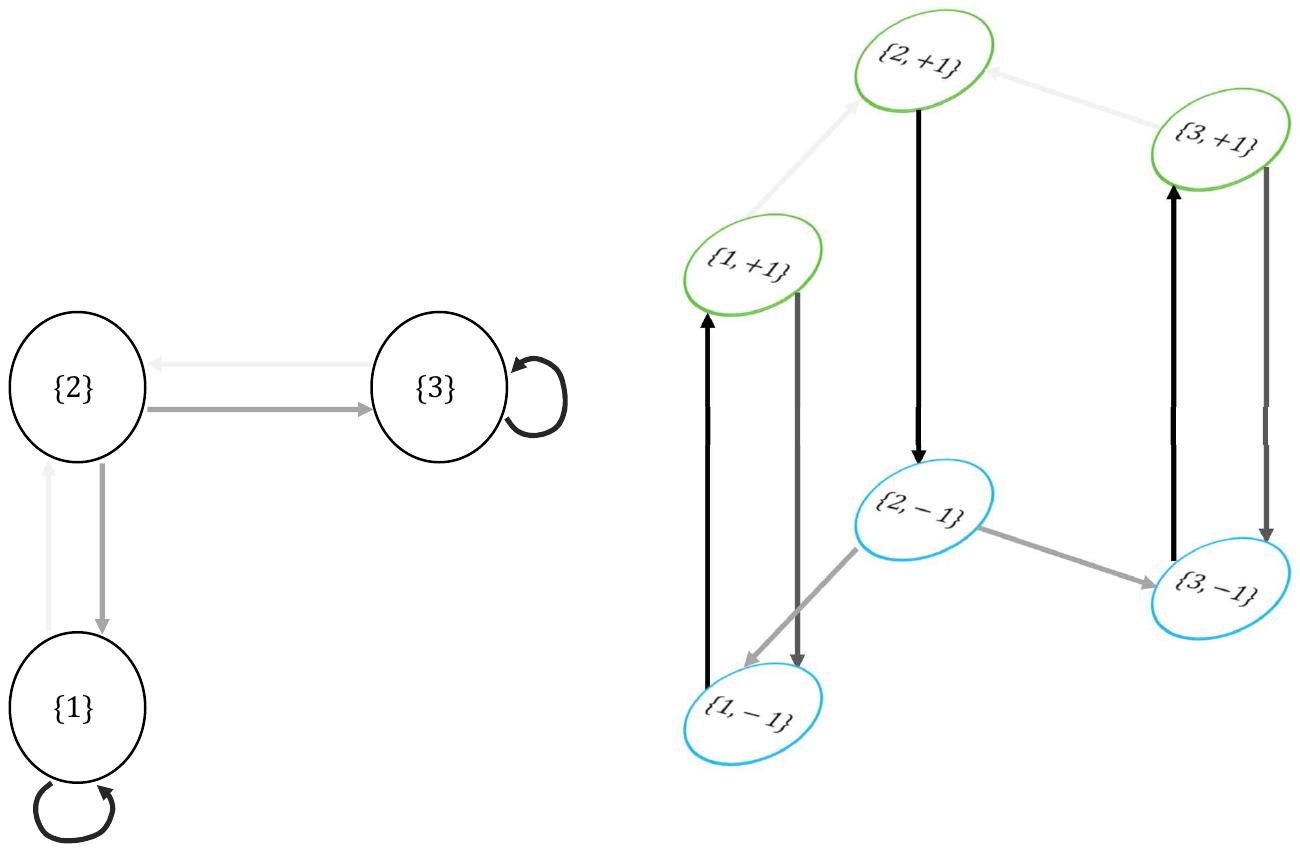}
  \vspace{-2mm}
\caption{Illustration of the MH chain (left) and its lifted counterpart (right); the colour of the arrows indicates the level of the transition probabilities (darker is higher).}\label{fig_example}
\end{figure}

\begin{Proposition}
In the context of this example, let $f$ be such that $f(1)=1=-f(3)$ and $f(2)=0$. Then, $\vara(P_{\lifted},f)=2\vara(P_{\MH},f)$ and $\vara(P_{\MH},f)$ is of order $1/\eps$.
\end{Proposition}

Given that the bound in \autoref{thm2} yields
\[
 2 = \frac{\vara(f, P_{\lifted})}{\vara(f, P_{\MH})} \leq 2 + \frac{\var[f(X)]}{\vara(f, P_{\MH})} = 2 + \frac{1 - \eps}{\vara(f, P_{\MH})},
\]
we know that it is essentially not possible to improve the bound as the term $(1 - \eps) / \vara(f, P_{\MH})$ on the right-hand side vanishes as $\epsilon \downarrow 0$. Before presenting the proof, we make two remarks about this example. Firstly, the asymptotic variance associated with the MH algorithm is of order $1/\eps$ because the states $\{1\}$ and $\{3\}$, that is the states that contribute to the ergodic average, are essentially isolated in the Markov chain (because $\pi\{2\} \ll \pi\{1\}=\pi\{3\}$ when $\eps$ is small). Secondly, we can view this example as an extreme one where the lifted sampler essentially never keeps its direction, which intuitively explains why the bound $\vara(f, P_{\lifted}) \leq 2\vara(f, P_{\MH}) + \var[f(X)]$ is essentially attained. Note that, if instead the North-South-East-West neighbourhoods are split into North-East and South-West neighbourhoods, we have that $\vara(f, P_{\lifted}) = \frac{1 - 3\eps}{2 - 3\eps} \, \vara(f, P_{\MH})$. In other words, by splitting the neighbourhoods otherwise, we obtain a lifted sampler which is better than the MH algorithm (in contrast to being worst when the neighbourhoods were split into North-West and South-East neighbourhoods). This highlights that different choices of directions may lead to lifted samplers that are significantly different.

\begin{proof}
We want to calculate $\vara(f, P) = \lim_{T \rightarrow \infty} \var\left[\frac{1}{\sqrt{T}} \sum_{k=1}^T f(X_k)\right]$ for $P \in \{P_{\MH}, P_{\lifted}\}$. The Markov chain simulated by the MH algorithm proceeds with independent and identically distributed \textit{tours} as follows:
\begin{itemize}
\item starting from $\{2\}$, go to $\{i\}$ with probability $1/2$, $i\in\{1,3\}$,
\item stay at $\{i\}$ for a number of iterations having a geometric distribution with parameter $\xi$,
\item return to $\{2\}$ (the new tour starts at this iteration).
\end{itemize}
Therefore, each tour has a length of $1 + W$, where $W$ has a geometric distribution with parameter $\xi$, and over a tour, the sum of $f(X_k)$ is $0 + UW$, $U \sim \mathcal{U}\{-1, +1\}$.

Let us define $N_T$, the number of completed tours after $T$ iterations, and two independent sequences $\{U_n\}$ and $\{W_n\}$ of independent and identically distributed random variables with the same distributions as above. Considering, without loss of generality, that $N_T \geq 1$, that is $1 + W_1 + \ldots + 1 + W_{N_T} = N_T + \sum_{n = 1}^{N_T} W_i < T$ and that the Markov chain starts at $\{2\}$, we have that
\[
 \frac{\sum_{k=1}^T f(X_k)}{\sqrt{T}}\stackrel{\text{dist.}}{=}\frac{\sum_{n=1}^{N_T} U_n W_n + U_{N_T+1} (T - N_T+\sum_{n=1}^{N_T} W_n)}{\sqrt{T}}.
\]
Note that $N_T$ is a (random) function of $T$ which goes to infinity as $T \rightarrow\infty$ (with probability 1). We have that $\Prob(N_T \leq T / 2) =1$ given that each tour has a length of at least $2$. Let us define $\F_T := \{W_1, \ldots, W_{N_T +1}, N_T\}$.

 We have that
\begin{align*}
 \E\left[\frac{\sum_{n=1}^{N_T} U_n W_n + U_{N_T+1} (T - N_T+\sum_{n=1}^{N_T} W_n)}{\sqrt{T}}\right] &= \E\left[\frac{\sum_{n=1}^{N_T} W_n \E\left[U_n \mid \F_T \right]}{\sqrt{T}}\right] \cr
  &\quad+ \E\left[\frac{(T - N_T+\sum_{n=1}^{N_T} W_n) \E\left[U_{N_T+1} \mid \F_T \right] }{\sqrt{T}}\right] \cr
  &= 0,
\end{align*}
using that $\E\left[U_n \mid \F_T \right] = 0, n = 1, \ldots, N_T+1$. Therefore, $\var\left[\frac{1}{\sqrt{T}} \sum_{k=1}^T f(X_k)\right] = \E\left[\left(\frac{1}{\sqrt{T}} \sum_{k=1}^T f(X_k)\right)^2\right]$.

The way we calculate $\vara(f, P_{\MH})$ is by establishing that $(1 / \sqrt{T}) \sum_{k=1}^T f(X_k)$ has a limiting distribution which is normal with a mean of 0 and a certain variance. We will also prove that $\left(\frac{1}{\sqrt{T}} \sum_{k=1}^T f(X_k)\right)^2$ is uniformly integrable, which will allow to establish that the variance of the limiting normal distribution is equal to $\vara(f, P_{\MH})$ \citep[Theorem 6.2]{dasgupta2008asymptotic}.

We have that
\begin{align*}
 \frac{\sum_{k=1}^T f(X_k)}{\sqrt{T}}&\stackrel{\text{dist.}}{=}\frac{\sum_{n=1}^{N_T} U_n W_n}{\sqrt{T}} + \frac{U_{N_T+1} (T - N_T+\sum_{n=1}^{N_T} W_n)}{\sqrt{T}} \cr
&= \frac{\sum_{n=1}^{N_T} \frac{U_n W_n}{\sqrt{1 + \E[W]}}}{\sqrt{\frac{T}{1 + \E[W]}}} + \frac{U_{N_T+1} (T - N_T+\sum_{n=1}^{N_T} W_n)}{\sqrt{T}} \cr
&\stackrel{\text{dist.}}{\rightarrow} \norm(0, \var[U W] / (1 + \E[W])),
\end{align*}
by combining Slutsky's theorem and Exercice 3.4.6\footnote{This exercice consists of proving the following result. Let $X_1, X_2, \ldots$ be independent and identically distributed random variables with $\E[X_i] = 0$ and $\E[X_i^2] =: \sigma^2 \in (0, \infty)$, and let $S_n = X_1 + \ldots + X_n$. Let $\{N_n\}$ be a sequence of non-negative integer-valued random variables and $\{a_n\}$ a sequence of non-negative integers with $a_n \rightarrow \infty$ and $N_n / a_n \rightarrow 1$ in probability. Then, $S_{N_n} / (\sigma \sqrt{a_n})$ converges in distribution towards a standard normal.} in \cite{durrett2019probability}. Indeed, $0 \leq T - N_T+\sum_{n=1}^{N_T} W_n < 1 + W_{N_T+1}$, which implies that
\[
 \frac{U_{N_T+1} (T - N_T+\sum_{n=1}^{N_T} W_n)}{\sqrt{T}} \rightarrow 0, \quad \text{with probability 1.}
\]
To apply the result in Exercice 3.4.6 in \cite{durrett2019probability}, which implies that
\[
 \frac{\sum_{n=1}^{N_T} \frac{U_n W_n}{\sqrt{1 + \E[W]}}}{\sqrt{\frac{T}{1 + \E[W]}}} \stackrel{\text{dist.}}{\rightarrow} \norm(0, \var[U W] / (1 + \E[W])),
\]
 we need to prove that $\frac{T / (1 + \E[W])}{N_T} \rightarrow 1$ in probability. With probability 1,
\[
 (1+W_1) + \ldots + (1 + W_{N_T}) \leq T \leq (1+W_1) + \ldots + (1 + W_{N_T + 1}),
\]
which is equivalent to
\[
  \frac{1}{N_T} \sum_{n = 1}^{N_T} (1 + W_n) \leq \frac{T}{N_T} \leq \frac{1}{N_T} \sum_{n = 1}^{N_T + 1} (1 + W_n),
\]
showing that, as $T \rightarrow \infty$,
\[
 \frac{T / (1 + \E[W])}{N_T} \rightarrow 1, \quad \text{with probability 1,}
\]
provided that $(1 / N_T) \sum_{n = 1}^{N_T} (1 + W_n) \rightarrow 1 + \E[W]$ with probability 1. We now prove this. By the strong law of large numbers and given that $N_T \rightarrow \infty$ as $T \rightarrow \infty$ with probability 1, on a set of realizations that has probability 1, there exists a positive constant $T_0$ such that for all $T \geq T_0$, $N_T$ is large enough and such that
\[
 \left|\frac{1}{N_T} \sum_{n = 1}^{N_T} (1 + W_n) - (1 + \E[W])\right| < \epsilon, \quad \epsilon > 0,
\]
yielding the result.

Given that $\var[U W]=\esp[U_1^2]\esp[W_1^2] = \esp[W_1^2] =(2-\xi)/\xi^2$, we obtain that
$$
\vara(f, P_{\MH})=\frac{1}{1+1/\xi}\frac{2-\xi}{\xi^2}=\frac{2-\xi}{\xi(\xi+1)}=\frac{(2 - 3\eps)(1 - \eps)}{\eps}\,,
$$
provided that $\left(\frac{1}{\sqrt{T}} \sum_{k=1}^T f(X_k)\right)^2$ is uniformly integrable, which we now prove. We proceed by proving that
\[
 \sup_T \E\left[\left(\frac{1}{\sqrt{T}} \sum_{k=1}^T f(X_k)\right)^4\right] < \infty.
\]

We have that
\begin{align*}
 \E\left[\left(\frac{1}{\sqrt{T}} \sum_{k=1}^T f(X_k)\right)^4\right] &= \frac{\E\left[\left(\sum_{n=1}^{N_T} U_n W_n + U_{N_T+1} (T - N_T+\sum_{n=1}^{N_T} W_n)\right)^4\right]}{T^2} \cr
 &= \frac{\E\left[\sum_{n=1}^{N_T} (U_n W_n)^4 + (U_{N_T+1} (T - N_T+\sum_{n=1}^{N_T} W_n))^4\right]}{T^2} \cr
 &\quad + \frac{3 \E\left[\sum_{n \neq m}^{N_T} (U_n W_n)^2 (U_m W_m)^2 + \sum_{n = 1}^{N_T} (U_n W_n)^2 (U_{N_T+1} (T - N_T+\sum_{n=1}^{N_T} W_n))^2\right]}{N^2} \cr
 &\leq \frac{\E\left[\sum_{n=1}^{N_T+1} W_n^4 \right]}{T^2} + \frac{3 \E\left[\sum_{n \neq m}^{N_T+1} W_n^2 W_m^2\right]}{T^2} \cr
 &\leq  \frac{\sum_{n=1}^{T/2+1}\E\left[W_n^4 \right]}{T^2} + \frac{3\sum_{n \neq m}^{T/2+1} \E\left[ W_n^2 W_m^2\right]}{T^2} \cr
 &= \frac{(T/2+1)\E[W^4] }{T^2} + \frac{3 (T / 2) (T / 2 + 1) (2-\xi)^2/\xi^4}{T^2},
\end{align*}
using in the second line an argument similar as above when we proved that $\E\left[\frac{1}{\sqrt{T}} \sum_{k=1}^T f(X_k)\right] = 0$, in the third line the definition of the random variables $U_n$ and that $T - N_T+\sum_{n=1}^{N_T} W_n < 1 + W_{N_T+1}$,  and in the fourth line that $N_T \leq T / 2$ with probability 1.  We thus have that $\sup_T \E\left[\left(\frac{1}{\sqrt{T}} \sum_{k=1}^T f(X_k)\right)^4\right] < \infty$ and $\left(\frac{1}{\sqrt{T}} \sum_{k=1}^T f(X_k)\right)^2$ is uniformly integrable.

We now turn to the lifted sampler. The Markov chain simulated by the lifted sampler also proceeds with independent and identically distributed tours as follows:
$$
(2,+1),\,(2,-1),\, \underline{(i,-1),\, (i,+1)},\,\underline{(i,-1),\, (i,+1)},\,\cdots,\,\underline{(i,-1),\, (i,+1)},(2,+1), \ldots
$$
where each block $\{(i,-1),\, (i,+1)\}$ is repeated $W$ times, $W$ having a geometric distribution with parameter $\xi$. From $(2, +1)$, the chain moves to $(2, -1)$ (with probability 1). From there, the chain moves to $(i, -1)$ with probability $1/2$, $i\in\{1,3\}$,  and next to $(i, +1)$ (with probability 1). An exit from $\{i\}$ can only occur while at $(i, +1)$ and with the same probability as in the reversible Markov chain. In summary, in this case, the tours are exactly twice as long as with the MH algorithm, that is their length is $2(1 + W)$, and, over a tour, the sum of $f(\tilde{X}_k)$ is $2UW$, denoting by $\tilde{X}_k$ the iterates of the $\Xset$ component in the Markov chain simulated by the lifted sampler (to highlight a difference with the reversible case). Let $T' := 2T$ be the number of iterations in the lifted sampler (that we write as a function of the number of iterations in the MH algorithm). We thus have that
\[
 \lim_{T' \rightarrow \infty}\frac{\sum_{k=1}^{T'} f(\tilde{X}_k)}{\sqrt{T'}}\stackrel{\text{dist.}}{=}\lim_{T \rightarrow \infty}\frac{2\sum_{n=1}^{N_T} U_n W_n}{\sqrt{2T}}\stackrel{\text{dist.}}{=}\sqrt{2}\lim_{T \rightarrow \infty}\frac{\sum_{k=1}^T f(X_k)}{\sqrt{T}},
\]
which allows to conclude that $\vara(P_{\lifted},f)=2\vara(P_{\MH},f)$. Indeed, given that the random variable $\sum_{k=1}^{T'} f(\tilde{X}_k) / \sqrt{T'}$ produced by the lifted sampler is asymptotically equivalent to $\sqrt{2}$ times that produced by the MH algorithm, that is $\sum_{k=1}^T f(X_k) /\sqrt{T}$, all the arguments presented above carry over to the lifted case.
\end{proof}

\end{document}